\documentclass[journal]{IEEEtran}
\usepackage{amsmath,amsfonts,amssymb,amsbsy,bm,paralist,theorem,color}
\usepackage{graphicx,algorithmic,algorithm,relsize}
\usepackage{multicol}
\usepackage{multirow}
\usepackage{subfigure}
\usepackage{cite}
\newtheorem{lem}{Lemma}
\newtheorem{rem}{Remark}
\newtheorem{thm}{Theorem}

\newtheorem{prop}{Proposition}
\theoremstyle{definition}
\newtheorem{defn}{Definition}

\graphicspath{{fig/}}

\definecolor{orange}{RGB}{255,107,0}
\definecolor{green}{RGB}{0,160,20}

\begin{document}

\title{Energy-Efficient Resource Allocation for NOMA enabled MEC Networks with Imperfect CSI}
\author{Fang Fang,~\IEEEmembership{Member,~IEEE},
Kaidi Wang,~\IEEEmembership{Student Member,~IEEE},\\ Zhiguo Ding,~\IEEEmembership{Fellow,~IEEE}, and Victor C.M. Leung,~\IEEEmembership{Fellow,~IEEE}
\thanks{Fang Fang is with the Department of Engineering, Durham University, Durham DH1 3LE, U.K.(e-mail:
fang.fang@durham.ac.uk).}
\thanks{Kaidi Wang and Zhiguo Ding are with the Department of Electrical and Electronic Engineering, The University of Manchester, M13 9PL, UK (e-mail: kaidi.wang@postgrad.manchester.ac.uk, zhiguo.ding@manchester.ac.uk).}
\thanks{V. C. M. Leung is with the College of Computer Science and Software Engineering, Shenzhen University, Shenzhen 518060, China, and also with the Department of Electrical and Computer Engineering, the University of British Columbia, Vancouver, BC V6T 1Z4, Canada (e-mail: vleung@ieee.org).}
}\maketitle
\begin{abstract}
		The combination of non-orthogonal multiple access (NOMA) and mobile edge computing (MEC) can significantly improve the spectrum efficiency beyond the fifth-generation network. In this paper, we mainly focus on energy-efficient resource allocation for a multi-user, multi-BS NOMA assisted MEC network with imperfect channel state information (CSI), in which each user can upload its tasks to multiple base stations (BSs) for remote executions. To minimize the energy consumption, we consider jointly optimizing the task assignment, power allocation and user association. As the main contribution, with imperfect CSI, the optimal closed-form expressions of task assignment and power allocation are analytically derived for the two-BS case. Specifically, the original formulated problem is nonconvex. We first transform the probabilistic problem into a non-probabilistic one. Subsequently, a bilevel programming method is proposed to derive the optimal solution. In addition, by incorporating the matching algorithm with the optimal task and power allocation, we propose a low complexity algorithm to efficiently optimize user association for the multi-user and multi-BS case. Simulations demonstrate that the proposed algorithm can yield much better performance than the conventional OMA scheme but also the identical results with lower complexity from the exhaustive search with the small number of BSs.
\end{abstract}

\section{Introduction}
In the last decade, increasing applications and services such as virtual reality (VR), augmented reality (AR), autonomous vehicle and wireless healthcare in Internet of things (IoT) have emerged in the evolution of wireless communication networks. However, most devices (e.g., sensors and wearable devices) have limited communication and storage resources and finite processing capabilities, which cannot support utral-low-latency and high-reliable communications. As a result, multi-access edge computing (MEC) has been proposed as a promising solution to enhance the computing capability of mobile devices with computation-intensive and latency-critical tasks \cite{ACTutMEC2017,TTMEC2017}. The performance gains on latency and energy consumption reduction motivated researchers to seamlessly apply MEC into wireless communications \cite{YMaoMECSureveys2017,PMMEC2017}. 
 To further improve spectrum efficiency, both non-orthogonal multiple access (NOMA) uplink transmission and NOMA downlink transmission have been proposed to apply in MEC \cite{ZDing2018TCOM}. NOMA-MEC can achieve superior performance on latency and energy consumption reduction over the traditional OMA-MEC system. 
 This motivated researches to investigate the performance gain of NOMA MEC networks \cite{ZNingNOMAMEC2019TVT,MVaeziNOMA2019}. 
\subsection{Related Literature}
In MEC networks, the devices with computation-intensive tasks need to offload (download) partial/entire tasks (task results) to (from) the MEC server located in close proximity \cite{SBMECMag2014,PMMECMag2017}. Depending on the partitionability and dependence of tasks, the offloading models can be classified into \emph{binary offloading} and \emph{partial offloading}. In the \emph{binary offloading} model, the task cannot be partitioned and needs to be offloaded as an entire task to the MEC server \cite{JRMECTWC2018}. While in the \emph{partial offloading} model, the task can be partitioned into multiple tasks, and parts of them can be offloaded to the MEC server for remote executions, then the remaining tasks can be executed locally at mobile devices \cite{YWangMEC2016}. 
The offloaded computation tasks can be executed at the MEC server, usually the base station (BS). Then computation results can be downloaded from the BS to users \cite{SAMECSurvey2014,WShiIEEEIoT2016}.

Communication and computation resource optimization plays a significant role in improving the system performance of NOMA MEC networks, which attracts extensive researchers to conduct research works on NOMA MEC. There are two categories according to different objectives: 1. Task delay minimization \cite{ZDing2018WCLDely,YWuTVTMECNOMA2019,FangTWC2019,LQ2019Iot2019}; 2. Energy consumption minimization \cite{ZDing2018TVT,FWang2018TCOM,AKiani2018JIOT,YPanCL2018,SHanIoT2019,ZYEEGC2018,ZSongCL2018,QGEEDL2018GC,FangGCWS2019}; Regarding to a single input and single output (SISO) NOMA MEC system, the authors in \cite{ZDing2018WCLDely,ZDing2018TVT} proposed a hybrid NOMA transmission scheme to minimize delay and energy consumption by considering fully offloading tasks to the MEC server. In particular, the optimal expressions of offloading power and time allocation were derived for the hybrid NOMA system, where a user can first offload parts of its task by pure NOMA, then offload the remaining by OMA. Both \emph{partial offloading} \cite{YWuTVTMECNOMA2019} and \emph{binary offloading} \cite{LQ2019Iot2019,FWang2018TCOM} were considered to minimize the task delay in NOMA MEC networks. In \cite{YWuTVTMECNOMA2019}, the system overall delay was minimized by an efficient layered algorithm for the NOMA-enabled multi-access MEC network. Regarding to \emph{partial offloading}, the task completion time was minimized by the proposed bisection method based algorithm for multi-user NOMA enabled MEC networks \cite{FangTWC2019}, where the optimal task assignment and power allocation expressions were derived for the two-user case. Besides, \emph{binary offloading} was considered to minimize the maximum task execution latency by optimizing SIC ordering and computation resource for NOMA MEC networks \cite{LQ2019Iot2019}. 
Regarding the multiple antenna model, \emph{binary offloading} and \emph{partial offloading} were both studied in \cite{FWang2018TCOM}, where a Lagrangian-based algorithm and a greed method based algorithm were proposed to minimize the energy consumption. Furthermore, the future wireless network is expected to achieve massive connectivity, which requires each BS to serve a large number of mobile devices by providing remote executions in NOMA MEC networks \cite{AKiani2018JIOT,YPanCL2018}. The energy efficiency optimization problem was investigated in NOMA based MEC networks \cite{SHanIoT2019,ZYEEGC2018}. In addition, the energy consumption minimization problem was studied for heterogeneous NOMA based MEC networks \cite{ZSongCL2018}. From NOMA transmission perspective, there are two main applications of NOMA in MEC including NOMA uplink transmission and NOMA downlink transmission. NOMA uplink transmission indicates that multiple users transmit signals to one single MEC server by using NOMA principle, which can ensure that multiple users complete their offloading simultaneously. While NOMA downlink transmission indicates one user offloads its tasks to multiple MEC servers by using NOMA protocol. Most existing works focus on NOMA uplink transmission enabled MEC networks while only a handful research works investigated NOMA downlink transmission in MEC \cite{QGEEDL2018GC,XDiaoAccess2019D2D,FangGCWS2019}. However, the perfect channel state information (CSI) is difficult to obtain in practice. To fully exploit the benefit of NOMA downlink transmission in MEC, such as supporting multi-server, in this work, we mainly focus on the energy consumption in NOMA downlink transmission enabled MEC\footnote{To avoid the confusion of NOMA downlink transmission in MEC and the traditional downlink transmission, we use NOMA transmission instead of NOMA downlink transmission for the rest of this paper.}.


\subsection{Motivations and Contributions}
Since perfect CSI is challenging to obtain in practice, in this paper, we consider imperfect CSI and investigate the resource allocation including task assignment and offloading power allocation for a multi-user and multi-BS NOMA-MEC network. It is worth to mention that the proposed resource allocation scheme is centralized. Similar to the emerging cloud radio access network (C-RAN), all BSs are connected by one central unit. The global system information including the estimated CSI, the computation tasks communication resource at devices and computing resource at BSs can be obtained at the central unit via high-capacity fronthaul links such as fiber links \cite{TXTranTVT2018}. The resource optimization is implemented in a centralized manner. The centralized unit performs the proposed algorithm to make decisions for users and sends them to BSs, which will broadcast the decisions to all the associated users by one pilot sequence. Considering the imperfect CSI, we propose an energy-efficient resource allocation scheme for a multi-user multi-BS network. We first focus on a two-BS case and propose the closed-form solution to the energy minimization problem. Subsequently, we focus on the multi-BS scenario, and a low-complexity user association is proposed to group each user with two BSs. Thus the obtained closed-form solution can be applied in each group to minimize the energy consumption. The detail contributions are listed as follows:
\begin{itemize}
	\item In this paper, we consider a multi-BS NOMA-MEC network with imperfect CSI. We aim to minimize the energy consumption of the offloading and computation processes by optimizing the transmit power, target offloading data rate and task partition at users. To reduce the decoding complexity at the receivers in NOMA transmission, we first consider that each user can offload its tasks to two BSs. The energy consumption minimization problem is formulated as a probabilistic problem, which is nonconvex. To efficiently solve the problem, we first transform the probabilistic problem into a nonprobabilistic one. Specifically, the outage probability constraint is incorporated into the objective function by using non-central chi-square distribution approximations.
	\item The transformed problem is still nonconvex and challenging to solve. We propose an optimal solution, where a bilevel programming method is proposed to minimize the energy consumption of the offloading phase and computing phase at MEC servers. Specifically, the closed-form expressions of the transmit power allocation and task partition are derived by carefully studying and analyzing the monotonicity and convexity of the problem. The optimal solution is concluded in five cases, which significantly reduces the computation complexity of the proposed scheme.
	\item We obtain some significant insights from the derived optimal solution, which clearly demonstrates the relationships between the optimal offloading schemes, pure NOMA and OMA, in the multi-BS NOMA-MEC network with imperfect CSI. The energy consumption efficiency (ECE) of communicating and computation phases is proposed to present conditions of pure NOMA offloading and OMA offloading. These insights demonstrate practical applications of the proposed resource allocation scheme. 
	\item For a more practical scenario, the multi-BS case, we propose a low-complexity algorithm to deal with the user association. Specifically, by using the obtained closed-form solutions for the two-BS case, we design two-sided matching to group each user with two BSs in the multi-user and multi-BS network. The optimal power allocation and task partition can be applied in each group. The complexity of the matching-based user association algorithm is significantly reduced compared to the optimal solution obtained through exhaustive searching. It can be shown that for a small number of BSs (M = 3), the proposed algorithm will yield identical results from the exhaustive search.
	\end{itemize} 
\subsection{Organization}
The rest of the paper is organized as follows: In Section II, the system model with the imperfect channel model and problem formulation are introduced. The optimal solution with closed-form expressions is proposed in Section III. In Section IV, an efficient user association algorithm is introduced. Simulation results are presented in Section V, and Section VI concludes the paper.
	\begin{figure}[t]
	\centering
	\graphicspath{{./figures/}}
\includegraphics[width=0.85\linewidth]{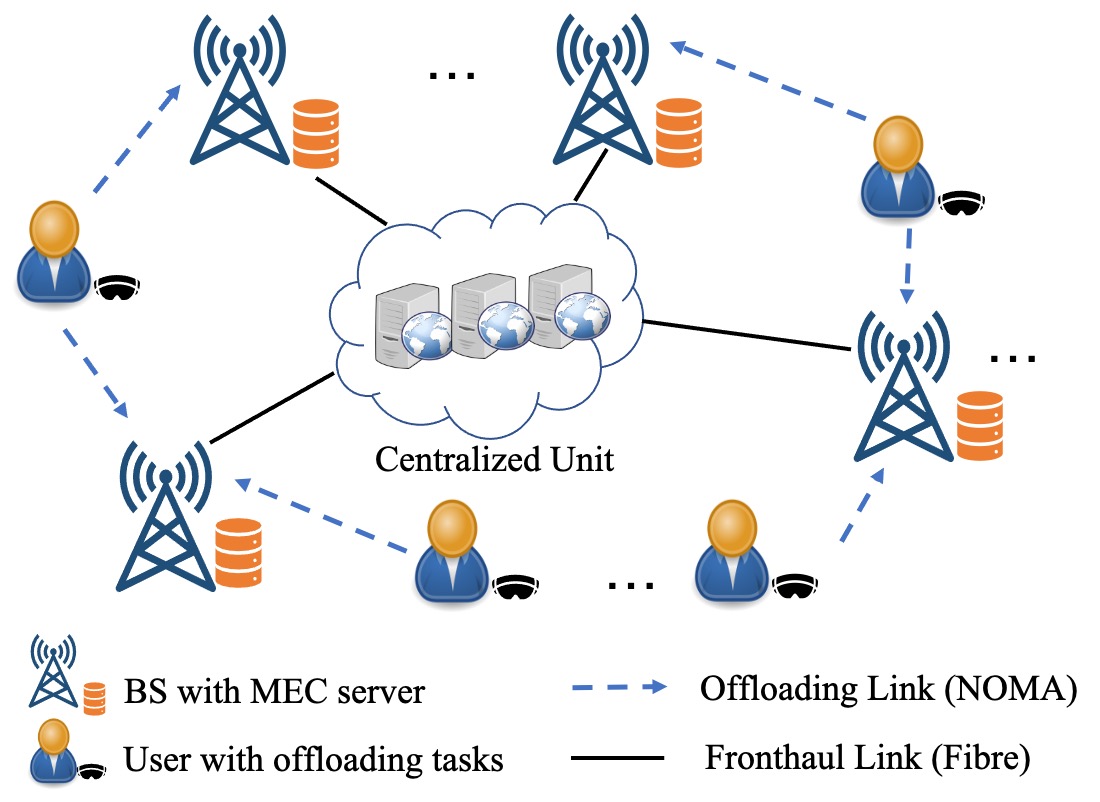}\\
	\caption{The multi-BS NOMA-MEC network.} \label{Fig0}
\end{figure}
\section{System Model and Problem Formulation}
\subsection{Multi-BS NOMA-MEC System Model} 
We consider a multi-BS NOMA-MEC network shown in Fig. \ref{Fig0}, where multiple users such as wearable devices and autonomous vehicles can offload their computation-intensive and latency-critical tasks to multiple BSs equipped with MEC servers. Each BS and each user are equipped with one antenna. To reduce energy consumption and task delay, each user requires to offload its tasks to multiple BSs in close proximity by NOMA transmission. This scenario is applicable in practice, especially for the cell edge users holding the computation-intensive task, but it has limited power to process the task by itself. The indices of $N$ users and $M$ BSs are respectively denoted by ${\rm UE}_m\in\{{\rm UE}_1,{\rm UE}_2,\cdots,{\rm UE}_M\}$ and ${\rm BS}_m\in\{{\rm BS}_1,{\rm BS}_2,\cdots,{\rm BS}_M\}$. In general, the task of each user can be described by $(L,C)$ where $L$ is the input number of bits for the task, and $C$ denotes the number of CPU cycles required to compute one bit of this task. In this system, we assume the task can be divided into several parts, and the user can offload different parts to different BSs for remote executions. The offloading task assignment ratio to ${\rm BS}_m$ is denoted by $\beta_{m}\in [0,1]$ and $\sum\limits_{m=1}^M\beta_{m}=1$. For example, there are two BSs. If the offloading bits to ${\rm BS}_1$ is $\beta_1L$, then the offloading bits of ${\rm BS}_2$ is $(1-\beta_1)L$.

We first focus the two-BS case and then extend it to the multi-BS case in Section IV. The channel gain from the user to ${\rm BS}_m$ is denoted by $g_m=h_m/d_m^{\frac{\alpha}{2}}$ where $h_m\sim \mathcal{CN}(0,1)$ is Rayleigh fading coefficients, and $d_m$ is the distance from the user to ${\rm BS}_m$, $\alpha$ is the path loss exponent, and where $\mathcal{CN}(0,1)$ is the complexed Gaussian distribution with mean zero and variance one. Without loss of generality, the channel gains of $M$ BSs are sorted as $|{g}_1|^2\leq|{g}_2|^2\leq\cdots\leq|{g}_M|^2$.
Assume that channel gains are constant within each transmission block, and vary from different blocks. The SIC technology is applied at BSs with a decoding order of decreasing order of the channel gains. Each BS can decode and remove the signals from BSs that have been decoded before. Denote the transmit power from the user to ${\rm BS}_m$ by $p_m$. Thus the signal received at ${\rm BS}_m$ is
	\begin{equation}
	\begin{aligned}
	y_m=|{g}_m|^2p_ms_m+\sum\limits_{i=m+1}^{M}|{g}_m|^2p_is_i+z
	\end{aligned}
	\end{equation}
	where $s_m$ is the transmit message to ${\rm BS}_m$, and $z\sim \mathcal{CN}(0,\sigma_{z}^2)$ is zero-mean additive white Gaussian noise (AWGN) with variance $\sigma_{z}^2$. The second term is the interference from other BSs. Define ${G}_m=\frac{|{g}_m|^2}{\sigma_{z}^2}$ as the channel gain ${G}_m$. Given by perfect CSI at BSs and the bandwidth $B$, the maximum achievable offloading data rate to ${\rm BS}_m$ can be written by
	\begin{equation}\label{C_m}
	\begin{aligned}
	C_m=B\log_2\left(1+\frac{{G}_mp_m}{\sum\limits_{i=m+1}^{M}{G}_mp_i+1}\right).
	\end{aligned}
	\end{equation}

\subsection{Imperfect CSI Channel Model}
Most previous works assumed that all the BSs know the entire knowledge of CSI. However, the perfect CSI is difficult to obtain in practice due to the high complexity of the back haul signalling overhead. In this paper, we investigate the energy consumption minimization by assuming that the small scale fading channel is estimated at BSs. The BSs forward the estimated CSI to the central unit via high-speed fronthaul links for the global decision making. In this section, the minimum mean square error (MMSE) channel estimation error model is adopted to describe the small scale fading coefficients $g_m$. Thus perfect channel gain can be written by 
\begin{equation}
g_m=\hat{g}_m+\epsilon
\end{equation} 
where $\hat{g}_m$ is the estimated channel gain including small scale fading $\hat{h}_m$ estimation and large scale fading $d^{\frac{\alpha}{2}}$, and $\epsilon\sim \mathcal{CN}(0,\sigma_{\epsilon}^2)$ is the channel estimation error with mean zero and variance $\sigma_{\epsilon}^2$. In this work, we assume that the large scale fading factors are perfectly estimated since the path loss and shadowing
are slowly varying. Thus we define the estimated channel gain from the user to ${\rm BS}_m$ normalized by $\sigma^2$ as $\hat{G}_m=\frac{\hat{h}_m}{\sigma^2}$. 

 Under imperfect CSI, a channel outage event happens when the instantaneous data rate with perfect CSI drops below the target rate. Define the target rate to ${\rm BS}_m$ as $R_m$. The actual channel gains $G_m$ are random variables since the estimate error $\epsilon$ is unknown. Given the target rate $R_m$, the outage probability can be defined as $\Pr\left[C_m<R_m|\hat{G}_m\right]$, which indicates 
 the communication from the user to ${\rm BS}_m$ fails when the instantaneous data rate $C_m$ drops below the target rate $R_m$ given by the estimated channel gain. To guarantee the quality of service (QoS) requirements, we usually limit the outage probability by $\varepsilon_o$. To corporate the outage probability into our system performance measurement, we adopt the average outage data rate \cite{NgTVT2010}
\begin{equation}
\begin{aligned}\label{Outage_R}
\hat{R}_m=R_m\Pr\left[C_m\geq R_m|\hat{G}_m\right] 
\end{aligned}
\end{equation}
where $\hat{R}_m$ indicates the minimum of the total average data rate successfully received by ${\rm BS}_m$. 

\subsection{Problem Formulation}
 In NOMA MEC, there are three phases to complete the task computation, i.e., task offloading, task computation at BSs and downloading task results from BSs to the user. In this work, we adopt the fully offloading scheme due to the limited power of the battery-powered devices \cite{ZDing2018TVT,ZDing2018WCLDely,ZDing2018TCOM}. The optimization is designed to save the processing energy at the MEC server and the energy consumption for transmitting tasks at users. The downloading transmission from the BSs to the user is not considered for the following two reasons: First, the size of task results are generally small \cite{CWangTVT2017,YHeTVT2018,YMaoJSAC2018}, and each BS has the more power to transmit the task result than that of the user. As a result, the downloading time is much shorter than offloading time. Second, the transmission energy consumption minimization at the BSs is more related to resource optimization from the perspective of BSs including transmit power allocation and subchannel allocation at BSs. However, in this work, we mainly focus on resource optimization from the perspective of users including offloading power, offloading data rate and task partition. Therefore, in this work, we investigate the energy consumption minimization including the transmitting energy consumption at users for transmitting signals to MEC servers and the computation energy consumption at the MEC servers in NOMA-MEC. 
\begin{itemize}
	\item \emph{Task offloading}: In this phase, the user offloads its partial task $\beta_m L$ to ${\rm BS}_m$. Given by the target data rate $R_m$, then the offloading time to ${\rm BS}_m$ is
	\begin{equation}
	T_m=\frac{\beta_mL}{R_m}.
	\end{equation}
	The energy consumption of offloading $\beta_m L$ task to ${\rm BS}_m$ is 
	\begin{equation}
	E_m=T_mp_m.
	\end{equation}
	\item \emph{Remote computation}: In this phase, the offloaded task will be computed by each BS's MEC server. Based on CPU frequency $f_m$ at ${\rm BS}_m$, the computing time at ${\rm BS}_m$ can be written by
	\begin{equation}
	T_m=\frac{\beta_{m}LC_m}{f_m}.
	\end{equation}  
	The computing energy consumption at ${\rm BS}_m$ is 
	\begin{equation}
	E_m^c=\kappa_m(1-\beta_m)L_mC_mf_m^2=\kappa_mT_mf_m^3
	\end{equation}
	where $\kappa_m$ denotes the effective capacitance coefficient for each CPU cycle of ${\rm BS}_m$.
\end{itemize}
According to the successive interference cancellation (SIC) technology applied at the receivers in NOMA protocol, the decoding complexity will exponentially increase with the number of receivers. To reduce the decoding complexity at BSs, we first consider the case that each user can only offload tasks to two BSs. For a multi-BS scenario, we propose a matching theory based algorithm for the user association shown in Algorithm 1 in Section IV. We aim to minimize the energy consumption of offloading and computing at MEC servers by optimizing the task assignment, target offloading rate and transmit power. The problem can be formulated as
\begin{subequations}\label{Prob:E_min1}
	\begin{align}
		\mathop {\min } \limits_{{\{\beta_1,\bm{R},\bm{p}\}}} \  &\frac{\beta_1 L}{R_1}p_1+\frac{(1-\beta_1) L}{R_2}p_2\\ \nonumber&+\kappa_{1}\beta_1LC_{1}f_{1}^{2}+\kappa_{2}(1-\beta_1)LC_{2}f_{2}^{2}\\
	\text{s.t.}   \quad &\Pr\left[C_m<R_m|\hat{G}_m\right]\leq \varepsilon_o,  m={1,2}, \label{Outage_con2}\\
	&0 \leq \beta_{1}\leq 1, \label{eq:beta2}\\ 
	&p_m\geq 0, m={1,2},\label{eq:p12}\\ 
	&p_1+p_2\leq P_{\max},\label{eq:p_sum2}\\ 
	&\frac{\beta_1 L}{R_1}\leq T_{\max},\label{eq:t1_range1}\\
	&\frac{(1-\beta_1) L}{R_2}\leq T_{\max}, \label{eq:t2_range1}\\
	&\frac{\beta_1 L}{R_1}=\frac{(1-\beta_1) L}{R_2}\label{eq:T_equal} \\ \nonumber
	\end{align}
\end{subequations}
where $\bm{R}=[R_1,R_2]^T$ and $\bm{p}=[p_1,p_2]^T$. Constraint \eqref{Outage_con2} limits the outage probability by $\varepsilon_o$;  constraint \eqref{eq:beta2} specifies the range of task assignment ratio; constraint \eqref{eq:p12} and constraint \eqref{eq:p_sum2} describe the rang of transmit power and the limitation of total transmit power; constraint \eqref{eq:t1_range1} and \eqref{eq:t2_range1} describe the delay limitations; constraint \eqref{eq:T_equal} guarantees that equal offloading time for all links (from the user to different BSs) via NOMA transmission. Note that problem \eqref{Prob:E_min1} is nonconvex and challenging to obtain the globally optimal solution in polynomial time.

  \section{Solution to Optimization Problem }
 Problem \eqref{Prob:E_min1} is nonconvex problem due to the outage constraint \eqref{Outage_con2}. In this section, we aim to incorporate the outage constraint \eqref{Outage_con2} into the objective function. In the following, we provide an effective way to transform the probabilistic problem to nonprobabilistic one.
 \subsection{Equivalent Data Rate with Imperfect CSI}
 By using the imperfect channel model, the actual instantaneous data rate can be rewritten as
 \begin{equation}
 	C_m=B\log_2\left(1+\frac{(\hat{G}_m+\epsilon)p_m}{\sum\limits_{i=m+1}^{M}(\hat{G}_m+\epsilon)p_i+1}\right).
 \end{equation}
 In general, the outage probability requirement is low (i.e., $\varepsilon_o \leq 0.1$). Thus the outage constraint can be satisfied with the equality at the optimal point \cite{NgTVT2010}. Therefore, we replace the ``$\leq$'' sign with a ``$=$'' sign in the following transformation \cite{SMImperfect2015}. The approximation is proved to be accurate \cite{Approxi}. As a result, the transformed optimization problem will be a more constrained version of problem \eqref{Prob:E_min1}. We introduce the following proposition to derive the target data rate:
 
 \begin{prop}\label{Out_RR}
 If we have the outage constraint:
 \begin{equation}\label{equalOut}
 	\Pr\left[C_m<R_m|\hat{G}_m\right]= \varepsilon_o,
 \end{equation}
 the target data rate can be derived as
 \begin{equation}
R_m=(1-\varepsilon_o)B\log_{2}\left(1+\frac{H_mp_m}{H_m\sum\limits_{i=m+1}^{M}p_i+\sigma_{z}^2}\right)
 \end{equation}
 where $H_m=-\ln(1-\varepsilon_o)2\left(1+\frac{\hat{g}_m}{\sigma_{\epsilon}^2}\right)/(\sigma_{z}^2d_m^{\alpha})$.	
 \end{prop}

 \begin{proof}
 The derivative proof can be found in Appendix \ref{ProofProp1}
 \end{proof}

 
\subsection{Problem Transformation}
Based on the target rate obtained from Proposition \ref{Out_RR}, the minimum average offloading data rate incorporated with outage constraint is $\tilde{R}_m=(1-\varepsilon_o)R_m$. Therefore, problem \eqref{Prob:E_min1} can be transformed as 
\begin{subequations}\label{Prob:E_min}
	\begin{align}
	\mathop {\min } \limits_{{\{\beta_1,\bm{p}\}}}\   &\frac{\beta_1 Lp_1}{\tilde{R}_1}+\frac{(1-\beta_1) Lp_2}{\tilde{R}_2}\\ \nonumber&+\kappa_{1}\beta_1LC_{1}f_{1}^{2}+\kappa_{2}(1-\beta_1)LC_{2}f_{2}^{2}\\
	\text{s.t.}   \quad 
	&0 \leq \beta_1 \leq 1 ,\label{eq:P>0}\\
&p_1 \geq 0, \quad p_2\geq 0,\label{eq:P0}\\
&p_1+p_2\leq P_{\max} \label{eq:Pm}.\\
&\frac{\beta_1 L}{\tilde{R}_1}\leq T_{\max},\label{eq:p1_range}\\
&\frac{(1-\beta_1) L}{\tilde{R}_2}\leq T_{\max}, \label{eq:p2_range}\\
&\frac{\beta_1 L}{\tilde{R}_1}=\frac{(1-\beta_1) L}{\tilde{R}_2}.\label{T_1=T_2}\\ \nonumber	
	\end{align}
\end{subequations}
This problem is nonprobabilistic problem, in which the outage constraint is incorporated into the target rate. However, this problem is still nonconvex. In order to obtain the globally optimal solution, we transform problem \eqref{Prob:E_min} into a programming problem \cite{YXuTSP2017}:
\begin{subequations}\label{Prob:Bilevel}
	\begin{align}
	\mathop {\min } \limits_{0\leq\beta_1\leq 1}\ g(\beta_1) \triangleq \mathop {\min } \limits_{{\{p_1,p_2\}}}\ &\frac{\beta_1 Lp_1}{\tilde{R}_1}+\frac{(1-\beta_1) Lp_2}{\tilde{R}_2}\\\nonumber &+\kappa_{1}\beta_1LC_{1}f_{1}^{2}+\kappa_{2}(1-\beta_1)LC_{2}f_{2}^{2}\\ 
	\text{s.t.} \quad &\eqref{eq:P0}-\eqref{T_1=T_2} \\
 \nonumber	
	\end{align}
\end{subequations}
where $g(\beta_1)$ is the inner problem. Given by fixed $\beta_1$, the inner problem $g(\beta_1)$ is the energy minimization problem with respect to $p_1$ and $p_2$, which is challenging to solve due to its nonconvexity. To solve this problem, we first analyze its monotonicity and obtain the following proposition.

\begin{prop} \label{monotonic}
	The energy consumption function is monotonic increasing with $p_1$ and $p_2$. Therefore, the minimum energy consumption is only achieved when the offloading power equals the minimum value:
\begin{subequations}\label{OptimalPower}
	\begin{align}
	&p_1^*=\left(2^{A\beta_1}-1\right)\left(\frac{1}{H_2}\left(2^{A\left(1-\beta_1\right)}-1\right)+\frac{1}{H_1}\right) \label{p11optimal}\\
      & p_2^*=\frac{2^{A(1-\beta_1)}-1}{H_{2}} \label{p12optimal}.
       \end{align}
    \end{subequations}
    where $A =\frac{L}{(1-\varepsilon_o)BT_{\max}}$.
\end{prop}
\begin{proof}
The proof can be found in Appendix \ref{PowerDerivation}
\end{proof}


\subsection{The Optimal Task Assignment Ratio Derivation}
\begin{figure*}
	\centering
	\graphicspath{{./figures/}}
	\subfigure[Strictly decreasing within feasible region.]{\begin{minipage}[b]{0.32\textwidth} 
\includegraphics[width=1\textwidth]{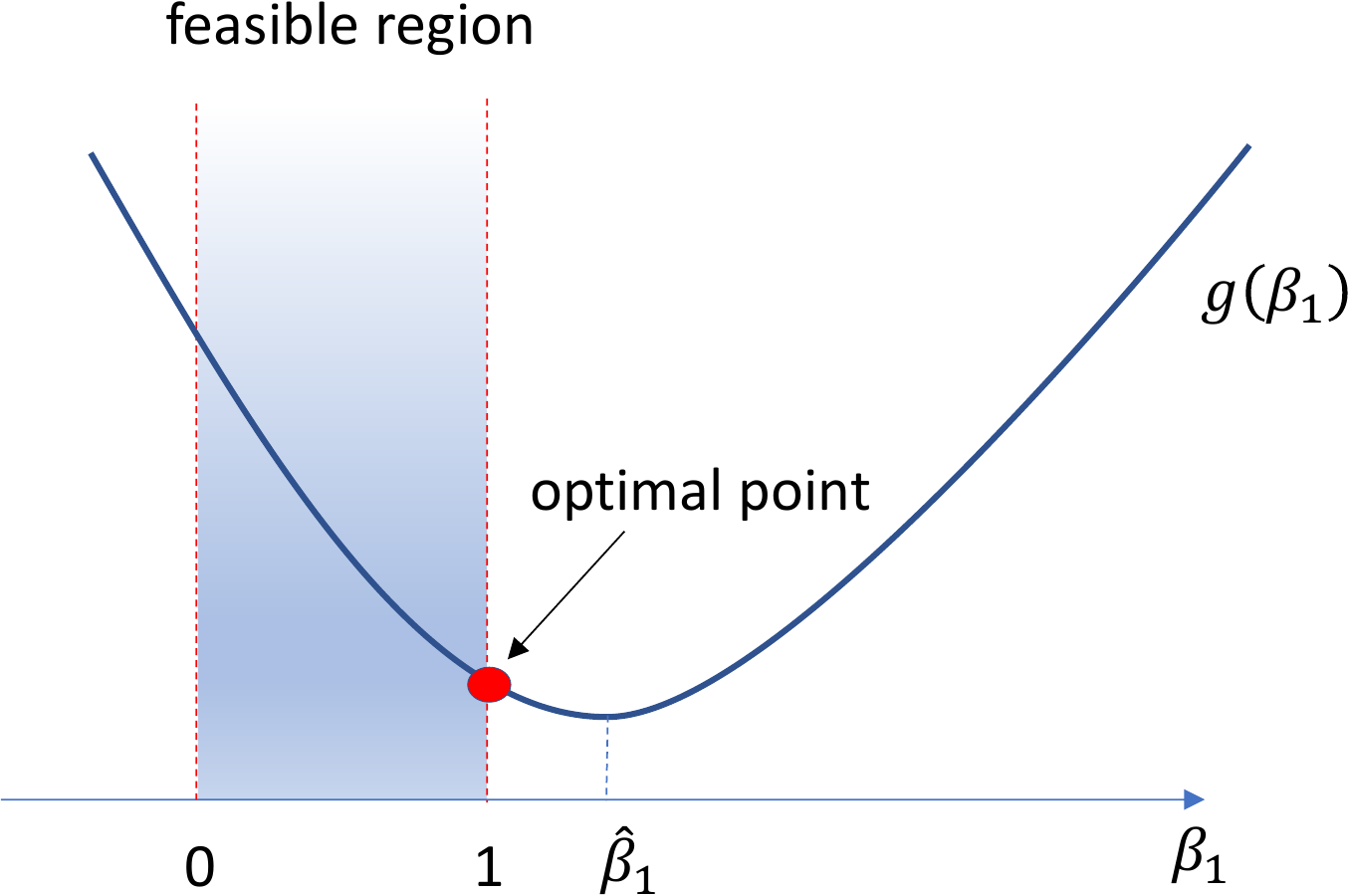} 
\end{minipage}}
\vspace{1em}
	\subfigure[First decreasing then increasing within feasible region.]{\begin{minipage}[b]{0.32\textwidth} 
	\graphicspath{{./figures/}}
\includegraphics[width=1\textwidth]{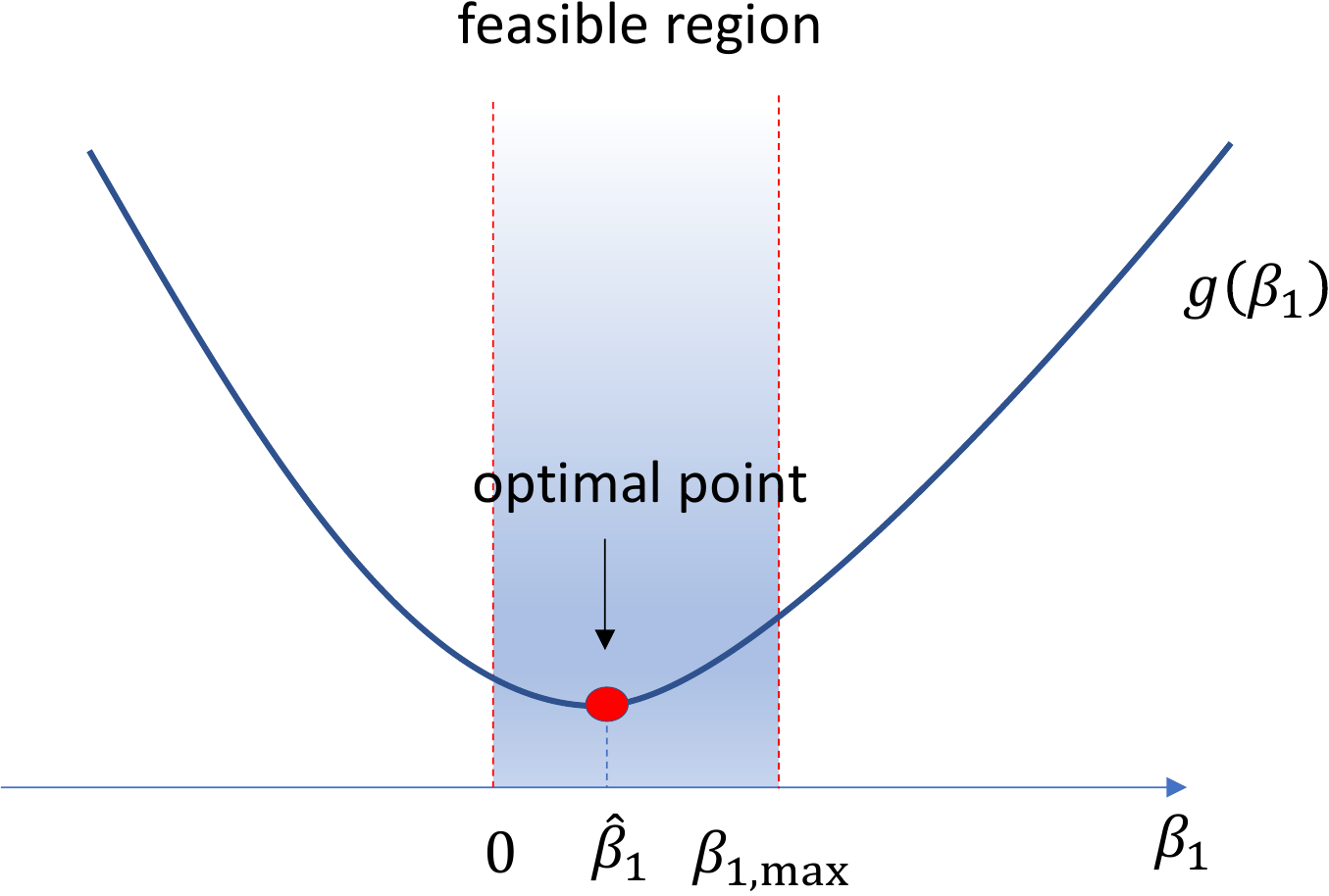} 
\end{minipage}}
	\subfigure[Strictly increasing within feasible region.]{\begin{minipage}[b]{0.32\textwidth} 
	\graphicspath{{./figures/}}
\includegraphics[width=1\textwidth]{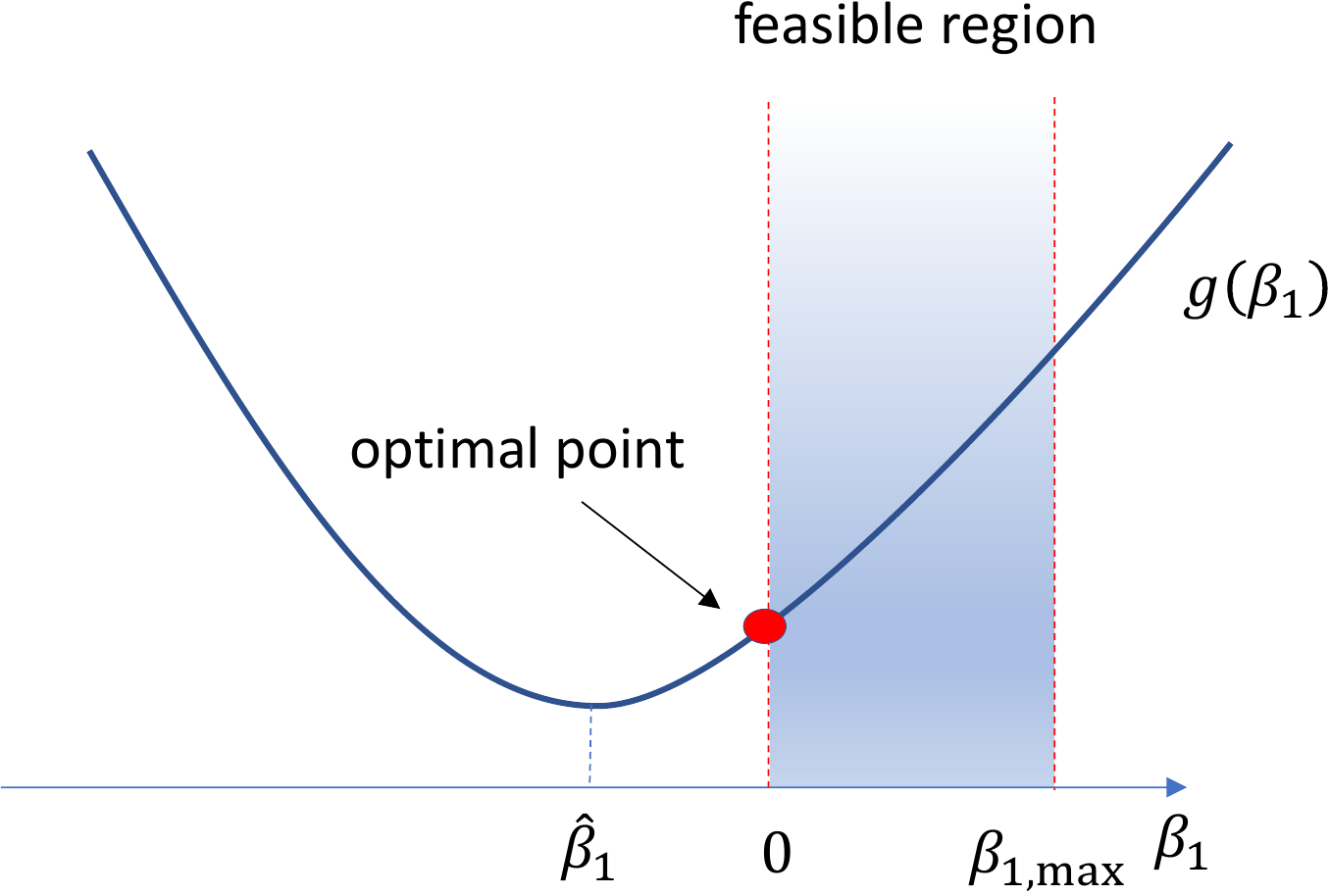} 
\end{minipage}}
\caption{Optimal task assignment ratio.}\label{Fig11}
	\end{figure*}
\begin{lem}\label{feasible}
In the following, we will derive the optimal task assignment ratio. Based on the optimal offloading power obtained from Proposition \ref{monotonic}, we have
\begin{equation}
\begin{aligned}
g\left(\beta_{1}\right)
=&T_{\textrm{max}}\left(\frac{1}{H_{2}}2^{A}+2^{A\beta_{1}}\left(\frac{1}{H_{1}}-\frac{1}{H_{2}}\right)-\frac{1}{H_{1}}\right)\\&+\left(\kappa_{1}LC_{1}f_{1}^{2}-\kappa_{2}LC_{2}f_{2}^{2}\right)\beta_{1}+\kappa_{2}LC_{2}f_{2}^{2}
\end{aligned}
\end{equation}
Thus the task assignment ratio optimization problem (the outer problem in \eqref{Prob:Bilevel}) can be rewritten by
\begin{subequations}\label{Prob:beta_1}
	\begin{align}
\underset{\beta_{1}}{\min}\ \ &g\left(\beta_{1}\right)\\
s.t.\ \ &0\leq\beta_{1}\leq1,\\
&\beta_{1}\leq\frac{1}{A}\log_{2}\left(\frac{P_{\max}+\frac{1}{H_{1}}-\frac{1}{H_{2}}2^{A}}{\frac{1}{H_{1}}-\frac{1}{H_{2}}}\right).
	\end{align}
\end{subequations}
To solve this problem, we have the following observations.
\begin{prop}
$g(\beta_1)$ is convex function with respective to $\beta_1$.
\end{prop}
\begin{proof}
\begin{equation}
\frac{\partial^{2}g}{\partial\beta_{1}^{2}}=\left(A\ln(2)\right)^{2}\left(\frac{1}{H_{1}}-\frac{1}{H_{2}}\right)2^{A\beta_{1}}\geq0.
\end{equation}	
\end{proof}

	The optimal solution of Problem \eqref{Prob:beta_1} relies on the feasible region of $\beta_1 \in [\beta_{1,\min},\beta_{1,\max}]$, which is
	\begin{equation}
	0\leq\beta_{1}\leq\max\left\{\hat{\beta}_{1,\max},1\right\} 
	\end{equation} 
	where 
	$\hat{\beta}_{1,\max}=\frac{1}{A}\log_{2}\left(\frac{P_{\max}+\frac{1}{H_{1}}-\frac{1}{H_{2}}2^{A}}{\frac{1}{H_{1}}-\frac{1}{H_{2}}}\right)$.
	To make Problem \eqref{Prob:beta_1} feasible, we must have
	\begin{equation}\label{feasibleCon}
	2^A\leq 1+H_2P_{\max}.
	\end{equation}
	This indicates that the offloading time must be no larger than maximum time delay $T_{\max}$ if all the tasks are transmitted to ${\rm BS}_1$.
	\begin{proof}
		To guarantee the feasible set, we must have 
		\begin{equation}
			\frac{1}{A}\log_{2}\left(\frac{P_{\max}+\frac{1}{H_{1}}-\frac{1}{H_{2}}2^{A}}{\frac{1}{H_{1}}-\frac{1}{H_{2}}}\right)\geq 0.
		\end{equation}
		Thus we have $2^A\leq 1+H_2P_{\max}.$
	\end{proof}
\end{lem}
To obtain the optimal solution, we have the following theorem for the analysis.
\begin{thm} \label{Theorem1}
	When Problem \eqref{Prob:beta_1} is feasible, due to  its convexity, the energy consumption $g(\beta_1)$
\begin{itemize}
\item [(a)] strictly decreases with $\beta_1$, when $\frac{\partial g}{\partial\beta_{1}}<0$ within the feasible region. 
\item[(b)] firstly strictly decreases and then strictly increasing within feasible region. $\frac{\partial g}{\partial\beta_{1}}<0$ when $\beta_1<\hat{\beta}_1$, and $\frac{\partial g}{\partial\beta_{1}}>0$ when $\beta_1>\hat{\beta}_1$ where
\begin{equation} \label{hat_beta}
\hat{\beta}_{1}=\frac{1}{A}\log_{2}\left(\frac{\kappa_{2}LC_{2}f_{2}^{2}-\kappa_{1}LC_{1}f_{1}^{2}}{A\ln(2)\left(\frac{1}{H_{1}}-\frac{1}{H_{2}}\right)}\right),
\end{equation}
which is achieved when $\frac{\partial g}{\partial\beta_{1}}=0$.
\item [(c)] strictly increases with $\beta_1$ within the feasible region.
\end{itemize}
\end{thm}

Theorem \ref{Theorem1} demonstrates the convexity of the energy consumption on the task assignment and guarantees the uniqueness of the globally optimal energy consumption, which can be illustrated by Fig. \ref{Fig11}.
\begin{figure*}
	\centering
	\vspace{1em}
	\subfigure[Optima task assignment ratio $\beta_1^*=\beta_{1,\max}$ ($\beta_{1,\max}$=1 or $\hat{\beta}_{1,\max}$).]{\begin{minipage}[b]{0.32\textwidth} 
	\graphicspath{{./figures/}}		\includegraphics[width=1\textwidth]{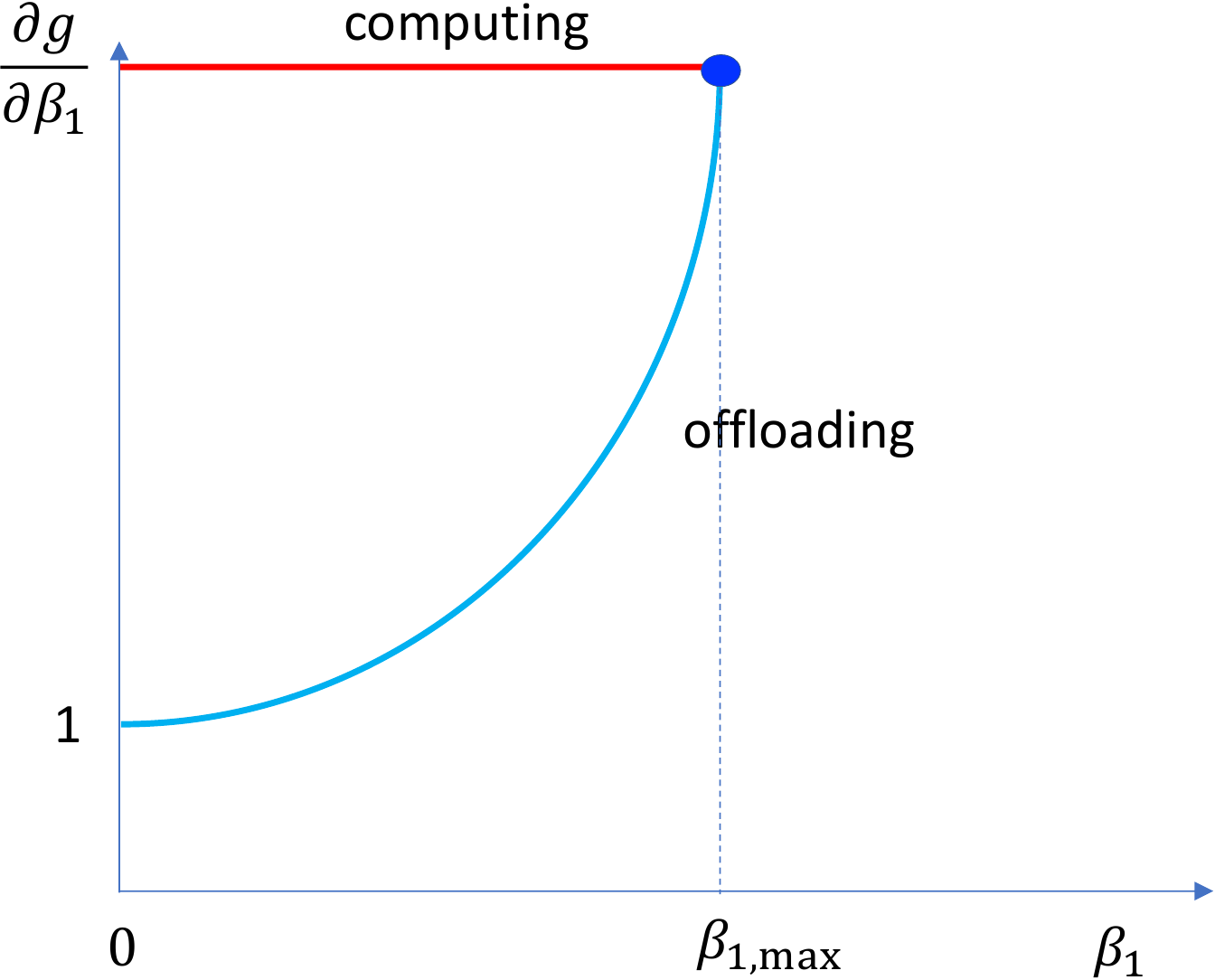} 
	\end{minipage}}
	\vspace{1em}
	\subfigure[Optima task assignment ratio $\beta_1^*=\hat{\beta}_1$.]{\begin{minipage}[b]{0.32\textwidth} 
		\graphicspath{{./figures/}}	\includegraphics[width=1\textwidth]{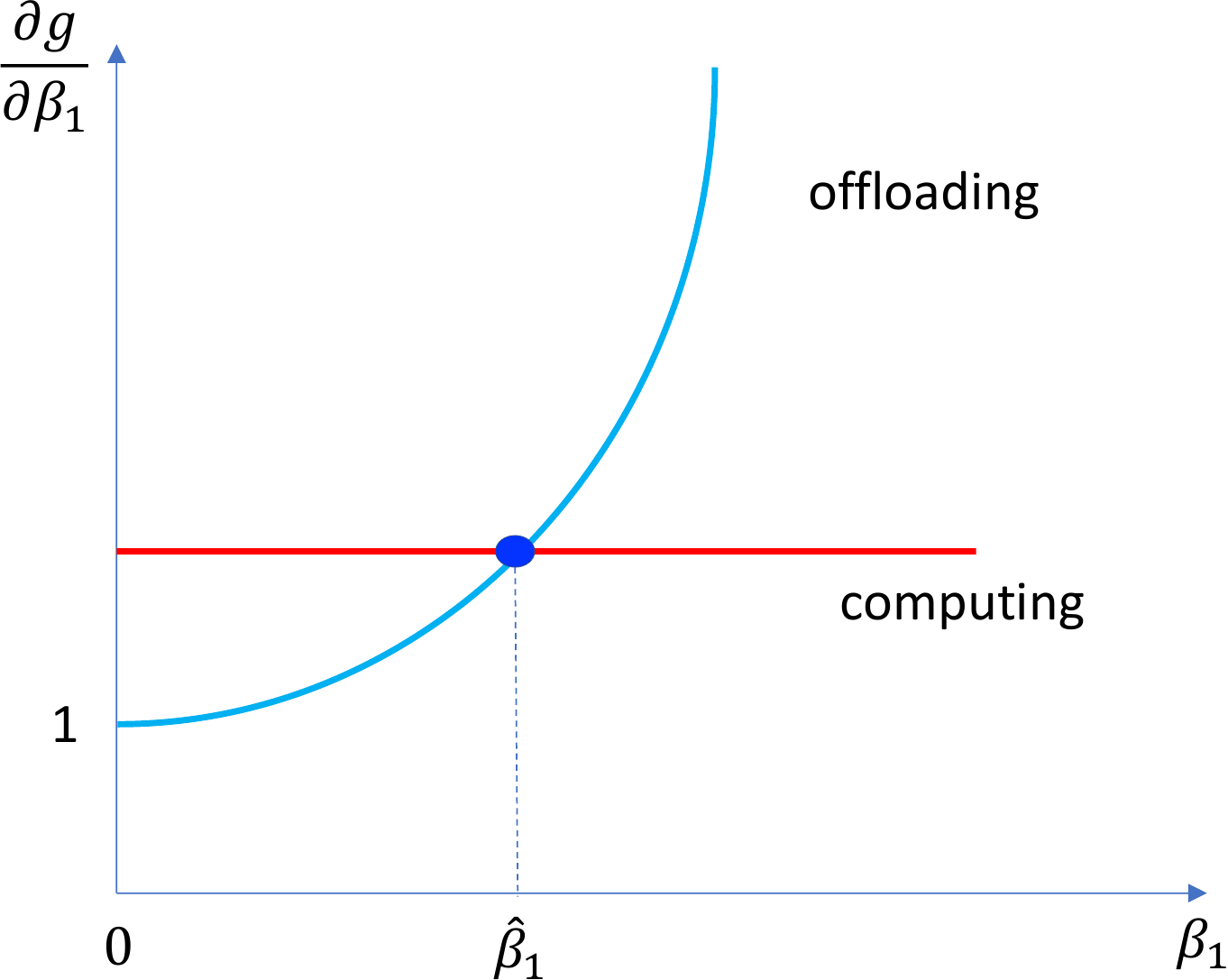} 
	\end{minipage}}
	\subfigure[Optima task assignment ratio $\beta_1^*=0$.]{\begin{minipage}[b]{0.32\textwidth} 
		\graphicspath{{./figures/}}	\includegraphics[width=1\textwidth]{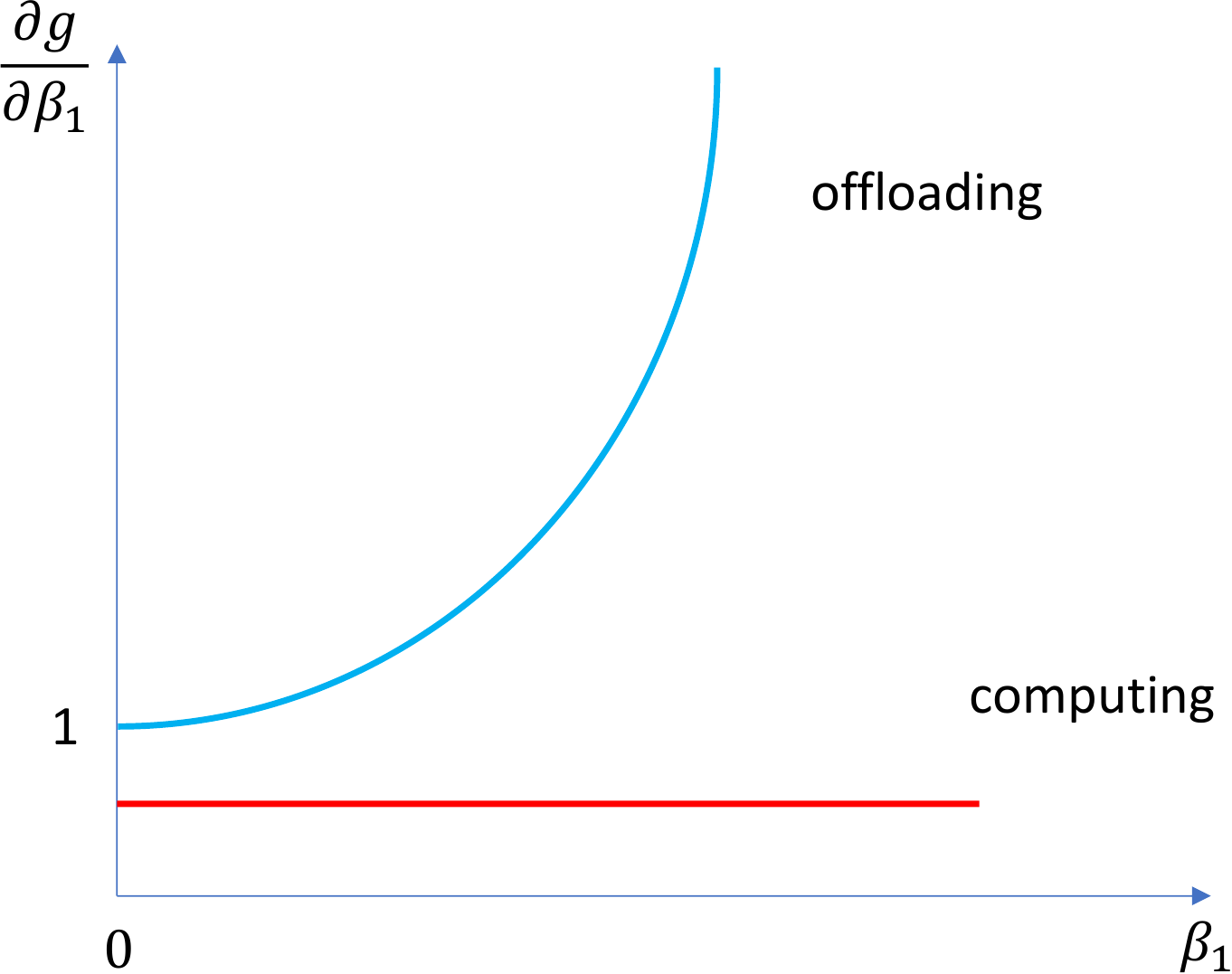} 
	\end{minipage}}
	\caption{Optimal task assignment ratio cases.}\label{Fig2}
\end{figure*}
In cases $(a)$ and $(b)$ of Theorem \ref{Theorem1}, the upper bound of $\beta_1$ can be 1 or $\hat{\beta}_{1,\max}$, then we have two optimal solutions for each case. In case $(c)$ of Theorem \ref{Theorem1}, there is only one optimal solution since the upper bound value will not affect the optimal solution. As a result, the optimal task assignment ratio can be concluded in the following five cases:  
\begin{itemize}
\item \emph{Case 1}: $\beta_1^*=1$, based on the conditions 
\begin{equation}\label{Case1Con}
\left\{
\begin{aligned}
&  2^A \leq 1+H_1P_{\max} \quad  & \\
&  2^A \leq \frac{\kappa_{2}LC_2f_2^2-\kappa_{1}LC_1f_1^2}{A\ln(2)\left(\frac{1}{H_1}-\frac{1}{H_2}\right)}. &
\end{aligned}
\right.
\end{equation}
In this case, the energy consumption is decreasing with $\beta_1$, which is illustrated by Fig. \ref{Fig11} (a). Thus the minimum energy consumption is achieved when $\beta_1^*=1$, which is OMA system, where the user only transmits its task to the ${\rm BS}_1$.

\item  \emph{Case 2}: $\beta_1^*=\hat{\beta}_{1,\max}$, based on the conditions 
\begin{equation}\label{Case2Con}
\left\{
\begin{aligned}
&  1+H_1P_{\max} \leq 2^A \leq 1+H_2P_{\max} \quad  & \\
& \frac{\kappa_{2}LC_2f_2^2-\kappa_{1}LC_1f_1^2}{A\ln(2)\left(\frac{1}{H_1}-\frac{1}{H_2}\right)} \geq \frac{P_{\max}+\frac{1}{H_{1}}-\frac{1}{H_{2}}2^{A}}{\frac{1}{H_{1}}-\frac{1}{H_{2}}}. &
\end{aligned}
\right.
\end{equation}
In this case, the energy consumption is decreasing with $\beta_1$. Since the upper bound of the feasible region is $\hat{\beta}_{1,\max}$, the minimum energy consumption can be achieved when $\beta_1^*=\hat{\beta}_{1,\max}$. This is pure NOMA offloading scheme, where the user offloads its tasks to two BSs simultaneously.
\item  \emph{Case 3}: $\beta_1^*=\hat{\beta}_1$, based on the conditions 
\begin{equation}\label{Case3Con}
\left\{
\begin{aligned}
&  2^A \leq 1+H_1P_{\max} \quad  & \\
&  1\leq \frac{\kappa_{2}LC_2f_2^2-\kappa_{1}LC_1f_1^2}{A\ln(2)\left(\frac{1}{H_1}-\frac{1}{H_2}\right)}\leq 2^A. &
\end{aligned}
\right.
\end{equation}
In this case, the energy consumption first decreases until $\beta_1$ reaches $\hat{\beta}_1$, and then increases with $\beta_1$ until $\beta_1$ reaches its maximum. Since the upper bound of the feasible region is $\hat{\beta}_{1,\max}$, the minimum energy consumption is achieved when $\beta_1^*=\hat{\beta}_{1,\max}$. This is pure NOMA offloading scheme.

\item  \emph{Case 4}: $\beta_1^*=\hat{\beta}_1$, based on the conditions 
\begin{equation}\label{Case4Con}
\left\{
\begin{aligned}
&  1+H_1P_{\max} \leq 2^A \leq 1+H_2P_{\max}   & \\
&  1\leq \frac{\kappa_{2}LC_2f_2^2-\kappa_{1}LC_1f_1^2}{A\ln(2)\left(\frac{1}{H_1}-\frac{1}{H_2}\right)}\leq \frac{P_{\max}+\frac{1}{H_{1}}-\frac{1}{H_{2}}2^{A}}{\frac{1}{H_{1}}-\frac{1}{H_{2}}}. &
\end{aligned}
\right.
\end{equation}
Similar with \emph{Case 4}, the upper bound of the feasible region is $1$. Thus the minimum energy consumption is achieved when $\beta_1^*=\hat{\beta}_{1,\max}$ with the above conditions. This is also pure NOMA offloading scheme.

\item  \emph{Case 5}: $\beta_1^*=0$, based on the conditions 
\begin{equation}\label{Case5Con}
\left\{
\begin{aligned}
&   2^A \leq 1+H_2P_{\max}   & \\
&  \frac{\kappa_{2}LC_2f_2^2-\kappa_{1}LC_1f_1^2}{A\ln(2)\left(\frac{1}{H_1}-\frac{1}{H_2}\right)}\leq 1. &
\end{aligned}
\right.
\end{equation}
In this case, the energy consumption increases with $\beta_1$, which is illustrated by Fig. \ref{Fig11} (c). Thus the minimum energy consumption is achieved when $\beta_1^*=0$. This is an OMA system, where the user only transmits its task to the ${\rm BS}_2$.

\end{itemize}
The detail derivation can be found in Appendix \ref{OptimalSolution}. Based on the optimal task assignment ratio, the optimal power allocation scheme can be achieved by closed-form expressions \eqref{OptimalPower}.

\subsection{Remarks and Discussions}
In this section, we present some analysis of the optimal solution for the energy minimization in NOMA transmission assisted MEC networks. Let us first define the energy consumption of each link (Link $m$ denotes the link from the user to ${\rm BS}_m$) with the task assignment ratio $\beta_1$ as 
\begin{equation}
\begin{aligned}
E_1\left(\beta_1\right)
=&T_{\textrm{max}}\left(2^{A\beta_1}-1\right)\left(\frac{1}{H_{2}}2^{A(1-\beta_1)}
+\frac{1}{H_{1}}-\frac{1}{H_{2}}\right)\\&+\kappa_{1}\beta_{1}LC_{1}\left(f_{1}\right)^{2}\\
E_2\left(\beta_1\right)=&T_{\textrm{max}}\frac{\left(2^{A\beta_1}-1\right)}{H_{2}}+\kappa_{2}\beta_1 LC_{2}\left(f_{2}\right)^{2}.
\end{aligned}
\end{equation}
Then we define energy consumption efficiency (${\rm ECE}_m$) of each link ($m=1,2$) by the derivatives, which includes energy consumption efficiency of offloading via the link to ${\rm BS}_m$ and computing at ${\rm BS}_m$. If  ${\rm ECE}_1>{\rm ECE}_2$, then Link 1 will consume more energy than Link 2 given by offloading task assignment ratio $\beta_1$
\begin{equation}\label{partialbeta1}
\begin{aligned}
{\rm ECE}_1=\frac{\partial E_1}{\partial\beta_1}=&\underset{{\rm ECE}_1^{o}}{\underbrace{\bar{A}\ln(2) 2^{A(1-\beta_1)}+\bar{A}\ln(2)\left(\frac{1}{H_{1}}-\frac{1}{H_{2}}\right)2^{A\beta_{1}}}}\\&+\underset{{\rm ECE}_1^{c}}{\underbrace{\kappa_{1}LC_{1}f_{1}^{2}}}\\
{\rm ECE}_2=\frac{\partial E_2}{\partial\beta_1}=&\underset{{\rm ECE}_2^{o}}{\underbrace{\bar{A}\ln(2) 2^{A(1-\beta_1)}}}+\underset{{\rm ECE}_2^{c}}{\underbrace{\kappa_{2}LC_{2}f_{2}^{2}}}\\
\end{aligned}
\end{equation}
where $\bar{A}=T_{\max}A$. ${\rm ECE}_m^o$ is the ECE of the offloading link to ${\rm BS}_m$. ${\rm ECE}_m^c$ is the ECE of the computing phase at ${\rm BS}_m$. To compare the ECE of these two links, we let
\begin{equation}\label{ECEDiff}
\begin{aligned}
{\rm ECE}_1-{\rm ECE}_2
=&\underset{{\rm ECE}{o}}{\underbrace{\bar{A}\ln(2)\left(\frac{1}{H_{1}}-\frac{1}{H_{2}}\right)2^{A\beta_{1}}}}-\\&\underset{{\rm ECE}{c}}{\underbrace{\left(\kappa_{2}LC_{2}f_{2}^{2}-\kappa_{1}LC_{1}f_{1}^{2}\right)}}
\end{aligned}
\end{equation}
where the first term, ${\rm ECE}o={\rm ECE}_1^o-{\rm ECE}_2^o$, is the ECE difference for the offloading phase. The second term, ${\rm ECE}c={\rm ECE}_1^c-{\rm ECE}_2^c$, is the ECE difference for computing phase. From the definition of ECE, we can obtain the following observations: ${\rm ECE}o$ is an increasing function of $\beta_1$.
\begin{rem}\label{rm1}
	The ECE of offloading phase, ECEo, is always positive based on $\frac{1}{H_1}-\frac{1}{H_2}\geq 0$. Thus, $ ECEo$ is a increasing function of $\beta_1$ due to $2^{A\beta_1}\geq 0$. 
\end{rem}
\begin{rem} \label{rm2}
	The ECE of computing phase, ECEc, is a constant. ECEc is positive when $BS_1$ has lower computing ECE than $BS_2$. Otherwise, ECEc is negative. 
\end{rem}
Based on Remarks \ref{rm1} and \ref{rm2}, the optimal solution can be concluded by two offloading schemes OMA system and pure NOMA system. The optimal solution can be interpreted by the ECE concept:

\begin{itemize} 
	\item [\emph{(1) ${\rm ECE}_1\leq {\rm ECE}_2, ({\rm ECE}_{\max}^o\leq {\rm ECE}c)$}]: this case corresponds to the scenario where the computing ECE difference ${\rm ECE}c$ is higher than the maximum of offloading ECE difference, ${\rm ECE}o$. This scenario is illustrated in Fig. \ref{Fig2}. (a). The upper bound of $\beta_1$ can be 1 or $\hat{\beta}_{1,\max}$, which correspond to the optimal solutions \emph{Case 1} and \emph{Case 2}, respectively. This scenario indicates that the ECE of the link to ${\rm BS}_1$ is far less than that of the link to ${\rm BS}_2$. In this case, the user prefers to offload its tasks to ${\rm BS}_1$ for remote executions to achieve the minimum energy consumption. 
	
	\item [\emph{(2) ${\rm ECE}_1= {\rm ECE}_2, ({\rm ECE}_{\min}^o<{\rm ECE}c<{\rm ECE}_{\max}^o)$}]: this case corresponds to the scenario in which the computing ECE difference, ${\rm ECE}c$, is higher than the minimum of offloading ECE difference, ${\rm ECE}_{\min}^o$ and lower than the maximum of offloading ECE, ${\rm ECE}_{\max}^o)$. This scenario is illustrated in Fig. \ref{Fig2}. (b). The upper bound of $\beta_1$ can be 1 or $\hat{\beta}_{1,\max}$, which correspond to the optimal solution \emph{Case 3} and \emph{Case 4}, respectively. This scenario is a pure NOMA offloading system, in which the user offloads its partial task $\beta_1L$ to ${\rm BS}_1$ and offload the remaining task $(1-\beta_{1}^*)L$ to ${\rm BS}_2$.
	
	\item [\emph{(3) ${\rm ECE}_1\geq {\rm ECE}_2, ({\rm ECE}_{\min}^o\geq {\rm ECE}c)$}]: this case corresponds to the scenario in which the computing ECE difference, ${\rm ECE}c$, is less than the minimum of offloading ECE difference, ${\rm ECE}_{\min}^o$. This is illustrated in Fig. \ref{Fig2}. (c). This case corresponds to the optimal solution \emph{Case 5}. This scenario is an OMA system, in which the user prefers to offload all its task to ${\rm BS}_2$ for remote executions to achieve the minimum energy consumption. This also indicates that the ECE of the link to ${\rm BS}_2$ is far less than that of the link to ${\rm BS}_1$. 
	\end{itemize}

\section{User association for multi-BS and multi-user via Matching}
In previous sections, the optimal energy-efficient resource allocation is derived in closed form for the two-BS case. To make our solution more practical, in this section, we consider a general scenario, where multiple users and multiple BSs are located in one single cell. To avoid the extremely high complexity of the exhaustive search method, we proposed a low-complexity algorithm for user association via matching theory.
\subsection{Design of User Association Algorithm}
 Assume that each user associated with two BSs occupies one subchannel. Thus the interference between users can be ignored due to different resource blocks. Let $\mathcal{B}=\left\{S_1, S_2, \cdots,S_{C_2^M}\right\}=\left\{\{{\rm BS}_1, {\rm BS}_2\}, 
\{{\rm BS}_1,{\rm BS}_3\},\cdots,\{{\rm BS}_{M-1},{\rm BS}_M\}\right\}$ denote the set of all the subsets of two distinct BSs\footnote{Generally, the number of users $N$ is larger than the number of BSs $M$. Therefore, the number of the possible subset including two BSs is $C_2^M$, where $C_2^M$ is the number of all the possible subsets of two distinct elements of $M$ BSs.} and $L=|\mathcal{B}|$. Let $\mathcal{U}=\{{\rm UE}_1,{\rm UE}_2,\cdots,{\rm UE}_N\}$ denote a set of users. We consider user association as a two-sided matching process between a set of $N$ users and a set $L$ of BS pairs. Therefore, the user association problem via matching ($L=N$) can be defined as:
\begin{defn}\label{Def1}
	A two-sided matching $\mathbb{M}$ is a mapping between the user set $\mathcal{U}$ and the BS pair set $\mathcal{B}$, satisfying the following conditions	
	\begin{itemize}
		\item [(1)] $\mathbb{M}({\rm UE}_n)\in\mathcal{B},\ \mathbb{M}(S_l)\in\mathcal{U},\forall n,l$;		
	    \item [(2)] $|\mathbb{M}({\rm UE}_n)|=1$, $|\mathbb{M}(S_l)|=1, \forall n,l$;
		\item [(3)] $S_l=\mathbb{M}({\rm UE}_n)\Leftrightarrow {\rm UE}_n=\mathbb{M}(S_l), \forall n,l$.
	\end{itemize}
\end{defn}	
In Definition \ref{Def1}, condition (1) indicates that each user in set $\mathbb{U}$ can be matched with a BS pair in set $\mathcal{B}$, and each BS pair in $\mathcal{B}$ is matched with a user in $\mathbb{U}$; Condition (2) states that each BS pair can be matched with only one user in $\mathbb{U}$ and vice versa; Property (3) implies that if ${\rm UE}_n$ is matched with $S_l$, then $S_l$ should be matched with ${\rm UE}_n$.

According to Definition \ref{Def1}, user association optimization is formulated as a two-sided matching problem. We aim to minimize the total energy consumption of the system. We first establish a preference list of users. For any ${\rm UE}_n \in \mathcal{U}$, ${\rm UE}_n$ prefers the BS pair $S_l$ rather $S_{l'}$ can be expressed as
\begin{equation}
	(S_l, \mathbb{M})\succ_{{\rm UE}_n} (S_{l'},\mathbb{M}')\Leftrightarrow EC_{{\rm UE}_n}(\mathbb{M})<EC_{{\rm UE}_n}(\mathbb{M}')
\end{equation}
where $EC_{{\rm UE}_n}$ is the energy consumption for ${\rm UE}_n$ associated with BSs in $S_l$. In terms of BS pairs, $S_l$ prefers to match with ${\rm UE}_n$ rather than ${\rm UE}_{n'}$ is described as 
\begin{equation}
	({\rm UE}_n,\mathbb{M})\succ_{S_l}({\rm UE}_{n'},\mathbb{M}')\Leftrightarrow EC_{S_l}(\mathbb{M})<EC_{S_{l'}}(\mathbb{M}')
\end{equation}
where $EC_{S_l}(\mathbb{M})$ is the energy consumption of the BS pair $S_l$ matched with user ${\rm UE}_n=\mathbb{M}(S_l)$.
\begin{algorithm}[!t] \small
	\caption{Matching Based User Association Algorithm }\label{Alg1}
	\begin{algorithmic}[1]		
		\STATE Initialize $\mathbb{M}$ by randomly matching each user to the BS groups.\\
		{\setlength\parindent{-1.5em} \bf Swap Matching Phase: }\vspace{1mm}
		\REPEAT
		\STATE Each ${{\rm UE}_i}$ searches other user ${{\rm UE}_j},\forall j\neq i$ to form the user pair $({\rm UE}_i,{\rm UE}_j)$.
		\IF {the user pair $({\rm UE}_i,{\rm UE}_j)$ is a swap blocking pair,}
		\STATE Swap the matching pair.
		\STATE Update $\mathbb{M}=\mathbb{M}_i^j$.
		\ENDIF
		\UNTIL There is no swap-blocking pair in $\mathbb{M}$.
	\end{algorithmic}
\end{algorithm}

To guarantee all the users are well matched with BSs, we develop a matching algorithm with low complexity to achieve a stable solution. We adopt swap matching, which is mathematically described as 
\begin{defn}\label{def2}
	A swap matching is denoted by \\$\mathbb{M}_n^{n'}=\{\mathbb{M}\backslash\{({\rm UE}_n,S_l),\ ({\rm UE}_{n'},S_{l'})\}\ \cup\  \{({\rm UE}_{n'}, S_l),({\rm UE}_{n},S_{l'})\}\}$,\\ where ${\rm UE}_n=\mathbb{M}(S_l),\ {\rm UE}_{n'}=\mathbb{M}(S_{l'}),\ {\rm UE}_{n}=\mathbb{M}_n^{n'}(S_l'),\  \text{and}\ {\rm UE}_{n'}=\mathbb{M}_n^{n'}(S_l)$.
\end{defn} 
where ${\rm UE}_n$ and ${\rm UE}_{n'}$ switch the matched BS pairs while keeping other matched pair in the matching scheme invariant. In a swap operation, considering their own interests, the player might not be approved by other users. Thus we introduce the concept of swap-blocking pair and then we evaluate the conditions under which the swap operations can be approved.
\begin{defn}\label{def3}
	Given matching $\mathbb{M}$ and two users $({\rm UE}_n,{\rm UE}_{n'})$ with ${\rm UE}_n=\mathbb{M}(S_l)$ and ${\rm UE}_{n'}=\mathbb{M}(S_{l'})$, if there exists a swap matching $\mathbb{M}_n^{n'}$ such that the energy consumption of these two users gets a decrease, then the swap operation is approved, and $({\rm UE}_n,\ {\rm UE}_{n'})$ is called a swap-block pair.
\end{defn}
Definition \ref{def3} implies that there is a benefit by exchanging the matching user pair $({\rm UE}_n,{\rm UE}_{n'})$ and this operation will not hurt the benefit of the other users' energy consumption. In the matching process, a potential swap blocking pair might be arranged by the scheduler, the scheduler will check if these two users can benefit from exchanging their matched BSs. The users will keep performing approved swap operations until they
reach a stable status, which is known as {\it two-sided exchange stable} matching, which is defined as
\begin{defn}\label{def4}
$\mathbb{M}$ is a two-sided exchange stable matching if $\mathbb{M}$ is not blocked by any swap blocking pai	
\end{defn}

Based on the above definitions, we proposed Algorithm \ref{Alg1} to solve the user association problem. In this algorithm, we first initialize the matching scheme $\mathbb{M}$. In the swap matching process, each use will iteratively check whether there are swap blocking pairs in the current matching scheme  $\mathbb{M}$. If so, we swap the user pair and update the current matching scheme. The matching process will terminate when there is no swap blocking pair in the current matching.
\subsection{Complexity Analysis}
For each ${\rm UE}_n$, there exist $N-1$ possible ${\rm UE}_{n'},n'\neq n$ to do swapping, thus the complexity order is given by $\mathcal{O}(\frac{N(N-1)}{2})$.
Therefore, the total complexity is $\mathcal{O}(K^2)$. Compared to the optimal strategy using the exhaustive search, which has a complexity order of $\mathcal{O}(K!)$, the computational complexity of the proposed
swap matching based algorithm is dramatically decreased.

\section{Simulations}

\begin{figure}[t]
\centering
\begin{minipage}[t]{0.48\textwidth}
\centering
\graphicspath{{./figures/}}
\includegraphics[width=0.95\linewidth]{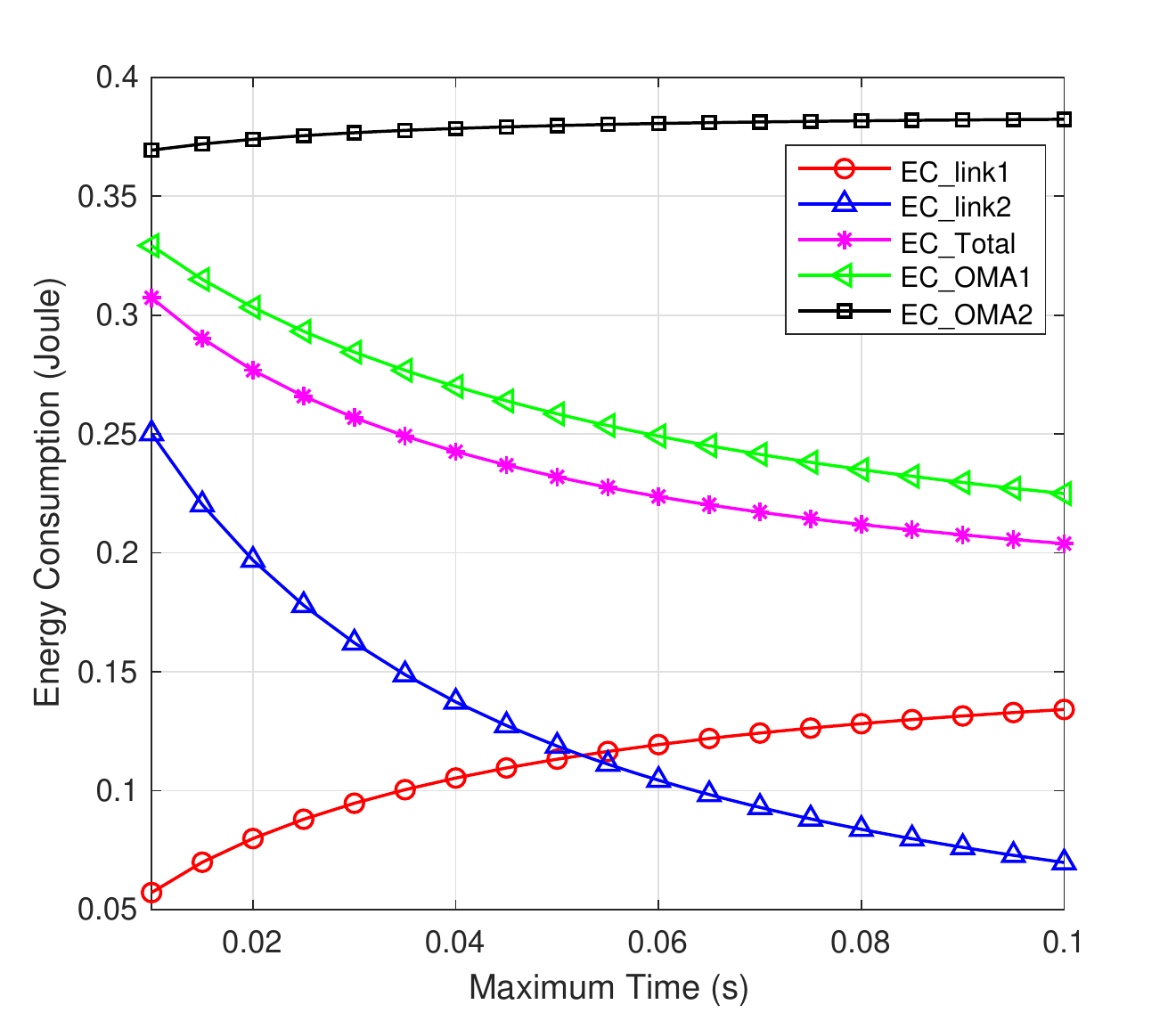}\\
	\caption{The energy consumption versus $T_{\max}$ with different offloading schemes. Radius: 500 m, $L=3.2\times10^6$ bits, $C=10^3\times[1,1]$, $\kappa=10^{-28}\times [0.8, 1.2]$ and $f=10^9\times[0.8,1]$.} \label{Fig3}
\end{minipage}
\begin{minipage}[t]{0.48\textwidth}
\centering
	\graphicspath{{./figures/}}\includegraphics[width=0.95\linewidth]{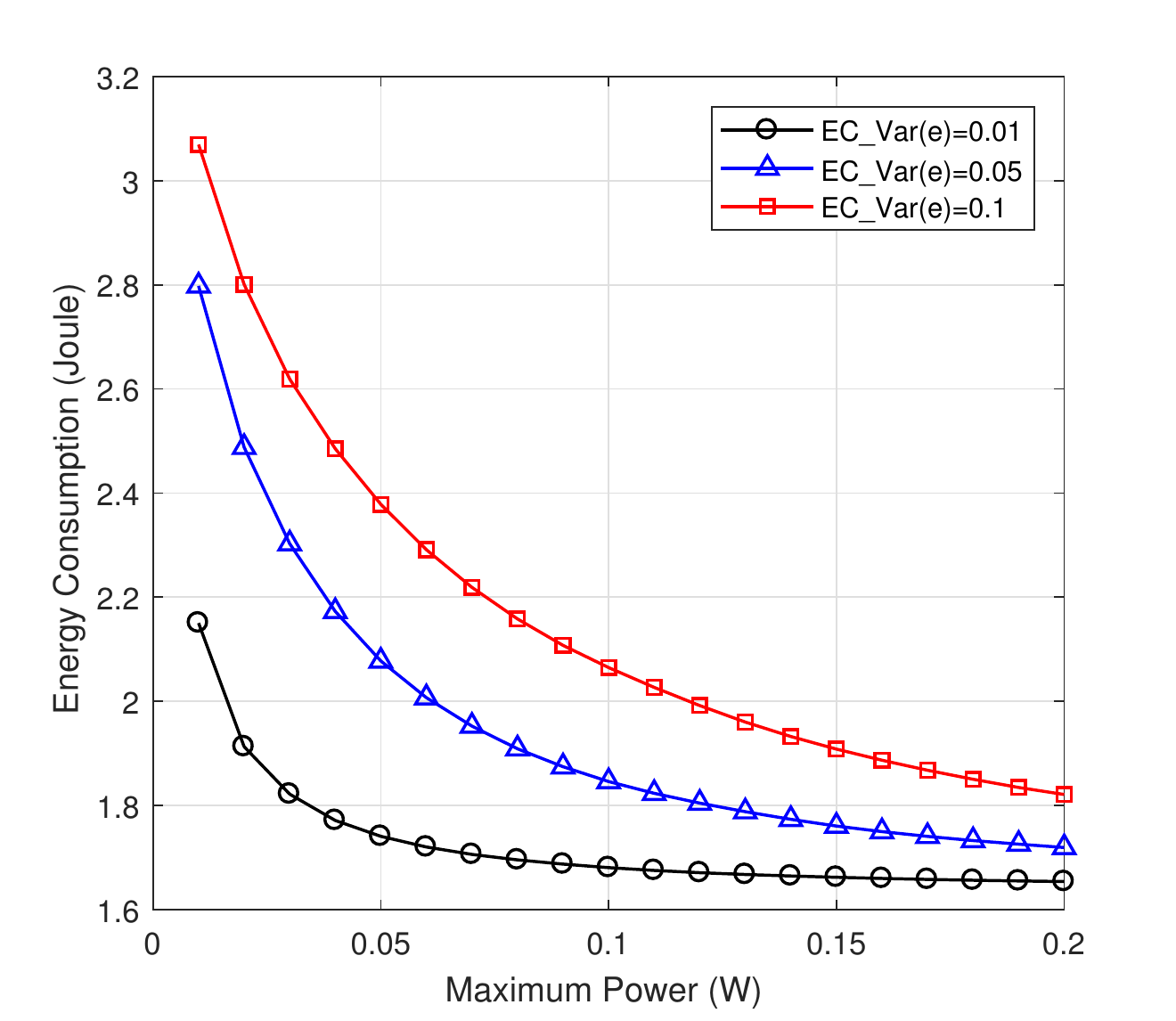}\\
	\caption{The energy consumption versus $P_{\max}$ with different variances of errors. Radius: 500 m, $L=3.2\times10^7$ bits, $C=10^3\times[1,1]$, $\kappa=10^{-28}\times [0.8, 1.2]$ and $f=10^9\times[0.8,1]$.} \label{Fig4}
\end{minipage}
\end{figure}

In this section, we evaluate the performance of our proposed scheme for the energy minimization in the downlink NOMA enabled MEC network. In the simulation setup, the user and two MEC servers are random distributed in a single disc with a radius of 500 m. We set the minimum distance between the user and BSs as 40 m to make our scenario more practical. The bandwidth is $B= 1$ GHz. According to the 3GPP urban path loss model, we set the path loss factor $\alpha=3.76$ \cite{3GPP}. The AWGN power is $\sigma_{z}^2=BN_0$ where $N_0=-174$ dBm/Hz is the AWGN spectral density. In order to guarantee the communication quality, the outage probability is set to $\varepsilon_o=0.1$.

In Fig. \ref{Fig3}, the energy consumption comparison of different offloading schemes is provided. We adopt two benchmarks for the performance comparison: 1. ${\rm BS}_1$ has the priority for offloading. 2. ${\rm BS}_2$ has the priority for offloading. We can observe that the system energy consumption decreases when the maximum time $T_{\max}$ increases. The offloading energy consumption of Link 1, including the offloading energy consumption on the Link 1 to ${\rm BS}_1$ and the computing energy consumption at ${\rm BS}_1$, decreases with $T_{\max}$. In the offloading scheme with the priority of ${\rm BS}_1$, the user will first offload its task to ${\rm BS}_1$. The link will be used when the task cannot be computed by ${\rm BS}_1$ within the $T_{\max}$. In this case, the fractional power allocation scheme \cite{FangJSAC17} is adopted to allocate different power to BSs. In other words, the link will be allocated more power when it has better channel gain than the other one. 
\begin{figure}[t]
\centering
\graphicspath{{./figures/}}
\begin{minipage}[t]{0.48\textwidth}
\centering
\includegraphics[width=0.95\linewidth]{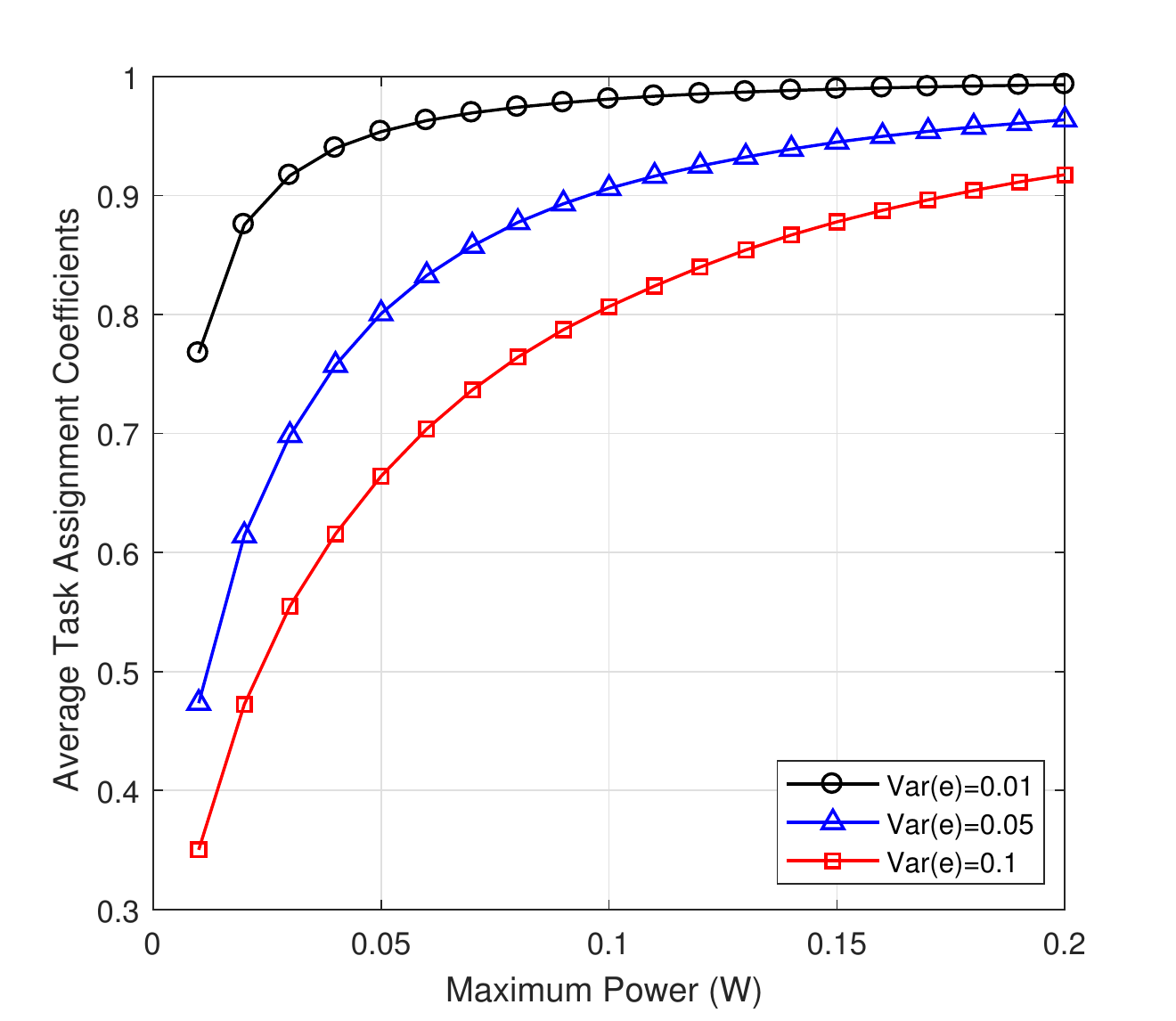}\\
	\caption{The average task assignment ratio to ${\rm BS}_1$ versus $P_{\max}$ with different error variances. Radius: 500 m, $L=3.2\times10^7$ bits, $C=10^3\times[1,1]$, $\kappa=10^{-28}\times [0.8, 1.2]$ and $f=10^9\times[0.8,1]$.} \label{Fig5}
\end{minipage}
\begin{minipage}[t]{0.48\textwidth}
\centering
\graphicspath{{./figures/}}	\includegraphics[width=0.95\linewidth]{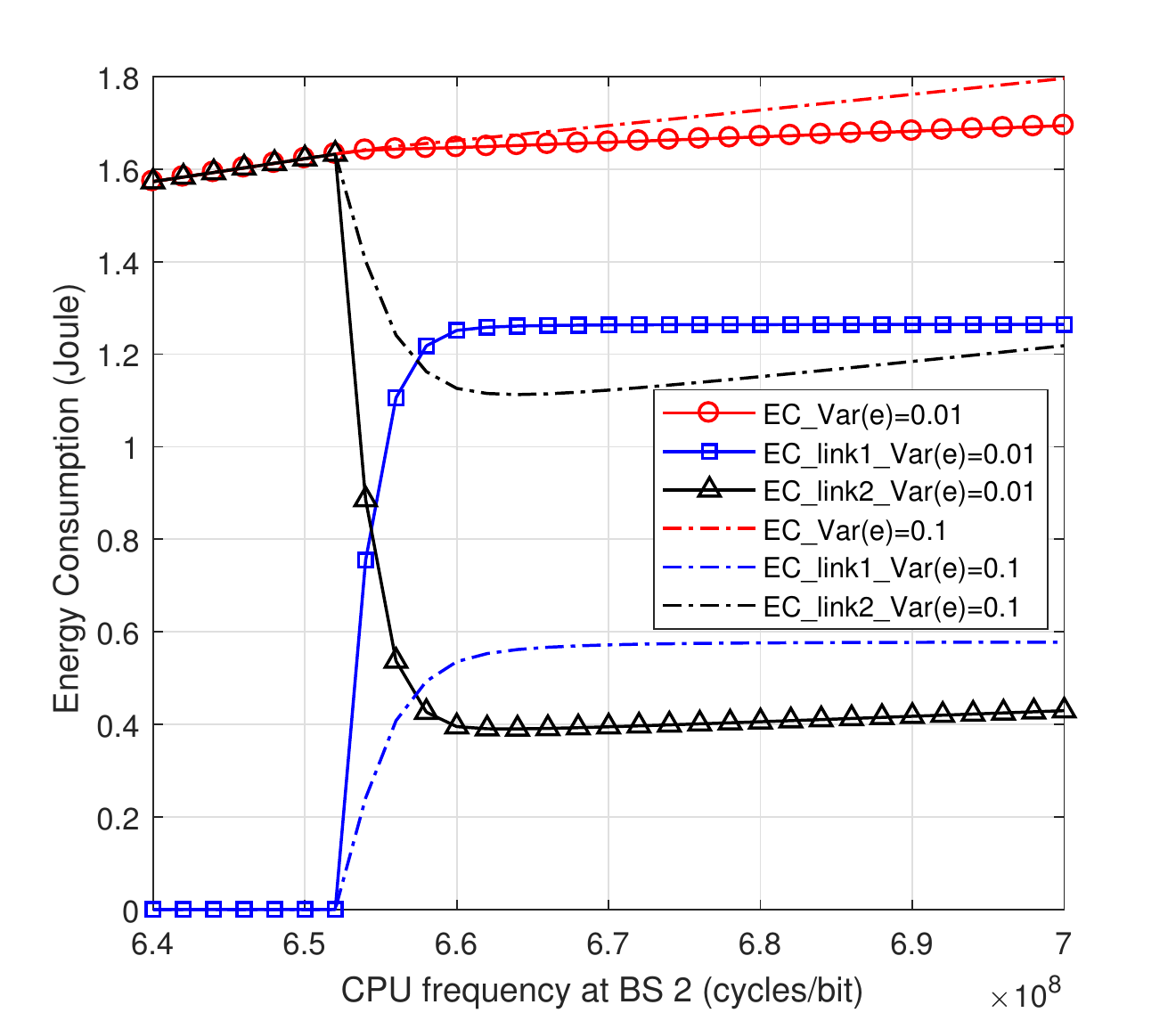}\\
	\caption{The energy consumption versus the CPU frequency $f_2$ at ${\rm BS}_2$ with different error variances. Radius: 500 m, $L=3.2\times10^7$ bits, $C=10^3\times[1,1]$, $\kappa=10^{-28}\times [0.8, 1.2]$ and $f_1=0.8\times10^9$.} \label{Fig6}
\end{minipage}
\end{figure}

Fig. \ref{Fig4} describes the energy consumption versus the maximum power with different variances of channel estimation error. The energy consumption decreases when the maximum power increases. As the maximum power grows larger, the energy consumption continues to decrease, but the decreasing rate becomes slower. This is because the feasible set of the optimization problem increases when the maximum power will increase. However, when it reaches a level, the optimal solution can be achieved without the effects of the maximum power. We call this maximum value the maximum power required to achieve the optimum. Moreover, Fig. \ref{Fig4} shows that the scheme with higher error variance will consume more energy than the scheme with lower error variance.  The task assignment ratio of Link 1 increases when the transmitted time increases.

Fig. \ref{Fig5} depicts the task assignment ratio to ${\rm BS}_1$ performance versus the maximum power. The task assignment ratio increases as the maximum power grows, but the increasing rate becomes slower. Moreover, Fig. \ref{Fig5} shows that the user prefers to assign more tasks to ${\rm BS}_1$ compared with the schemes with higher error variances. 
\begin{figure}[t]
\centering
\graphicspath{{./figures/}}
\begin{minipage}[t]{0.48\textwidth}
\centering
\includegraphics[width=0.95\linewidth]{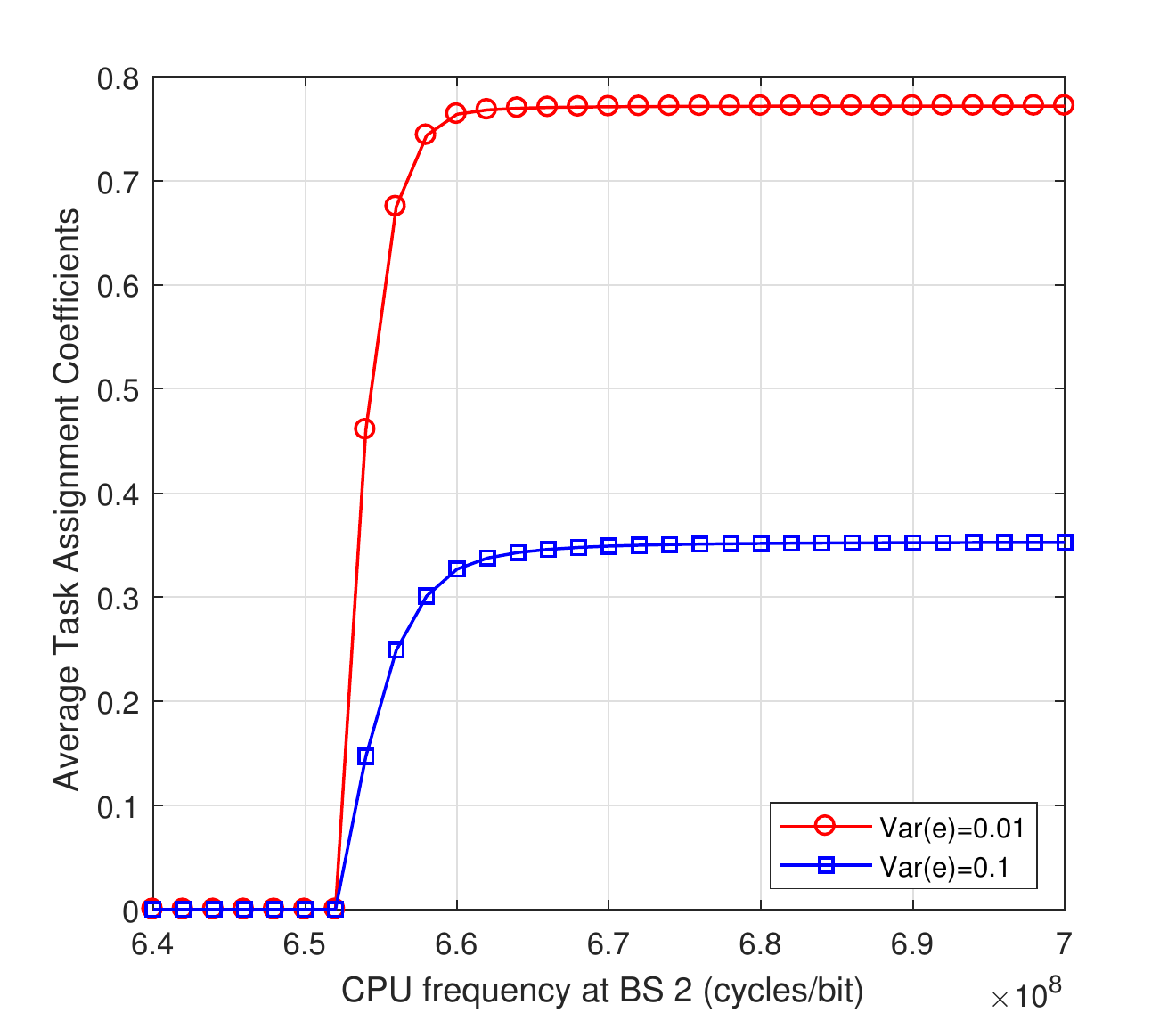}\\
	\caption{The average task assignment ratio to ${\rm BS}_1$ versus the CPU frequency $f_2$ at ${\rm BS}_2$ with different error variances. Radius: 500 m, $L=3.2\times10^7$ bits, $C=10^3\times[1,1]$, $\kappa=10^{-28}\times [0.8, 1.2]$ and $f_1=0.8\times10^9$.} \label{Fig7}
\end{minipage}
\begin{minipage}[t]{0.48\textwidth}
\centering
\graphicspath{{./figures/}}	\includegraphics[width=0.95\linewidth]{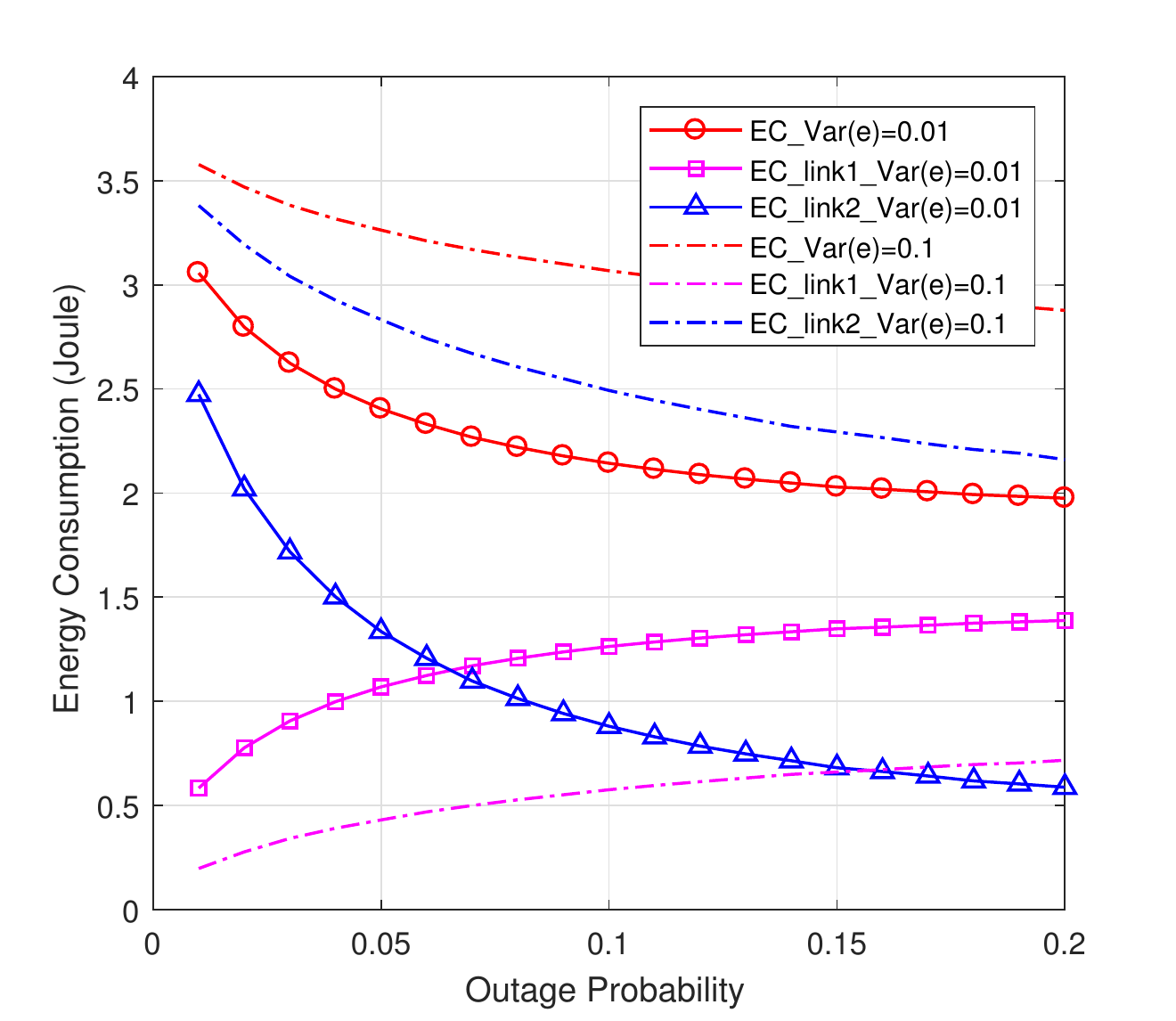}\\
	\caption{The energy consumption versus $T_{\max}$ with different offloading schemes. Radius: 500 m, $L=3.2\times10^7$ bits, $C=10^3\times[1,1]$, $\kappa=10^{-28}\times [0.8, 1.2]$, $f=10^9\times[0.8,1]$, $T_{\max}=0.01$ s and $P_{\max}=0.01$ W.} \label{Fig8}
\end{minipage}
\end{figure}
Fig. \ref{Fig6} shows the energy consumption versus the CPU frequency of ${\rm BS}_2$ by considering the schemes with different variances of channel estimation error. From Fig. \ref{Fig6}, the total energy consumption increases as the CPU frequency grows. This is because that the energy consumption increases as the $f_2$ grows. Moreover, the scheme with lower channel estimation error variance has lower energy consumption than the scheme with higher channel estimation error variance. Thus, as expected, the channel estimation error can degrade the energy efficiency performance.

Fig. \ref{Fig7} shows the average task assignment ratio to ${\rm BS}_1$ versus CPU frequency of ${\rm BS}_2$. We can observe that the offloading task assignment ratio is zero before the value of $f_2$, which corresponds to ${\rm ECE}c=0$. Before this point, the best solution is the OMA system. The user only offloads its task to ${\rm BS}_1$. After this point, the user starts to offload the partial task to ${\rm BS}_2$ as well. The task assignment ratio to ${\rm BS}_1$ increases as the CPU frequency of ${\rm BS}_2$ grows and stays stable after it reaches its optimal value. This corresponds to the pure NOMA offloading scheme. It can also be observed that the scheme with lower error variance has a higher task assignment ratio than the scheme with a higher estimated error variance.

Fig. \ref{Fig8} evaluates energy consumption versus the outage probability. It can be seen that the energy consumption decreases when the outage probability increases. When the outage probability is low, a high target rate is required to guarantee the simultaneous communication has a high probability of being successful. In this case, a higher target rate must cost more energy than a lower target rate. Similarly, a high outage probability requires a lower target rate, which costs lower energy than hither target rate.

\begin{figure}[t]
\centering
\begin{minipage}[t]{0.48\textwidth}
\centering
	\graphicspath{{./figures/}}\includegraphics[width=0.95\linewidth]{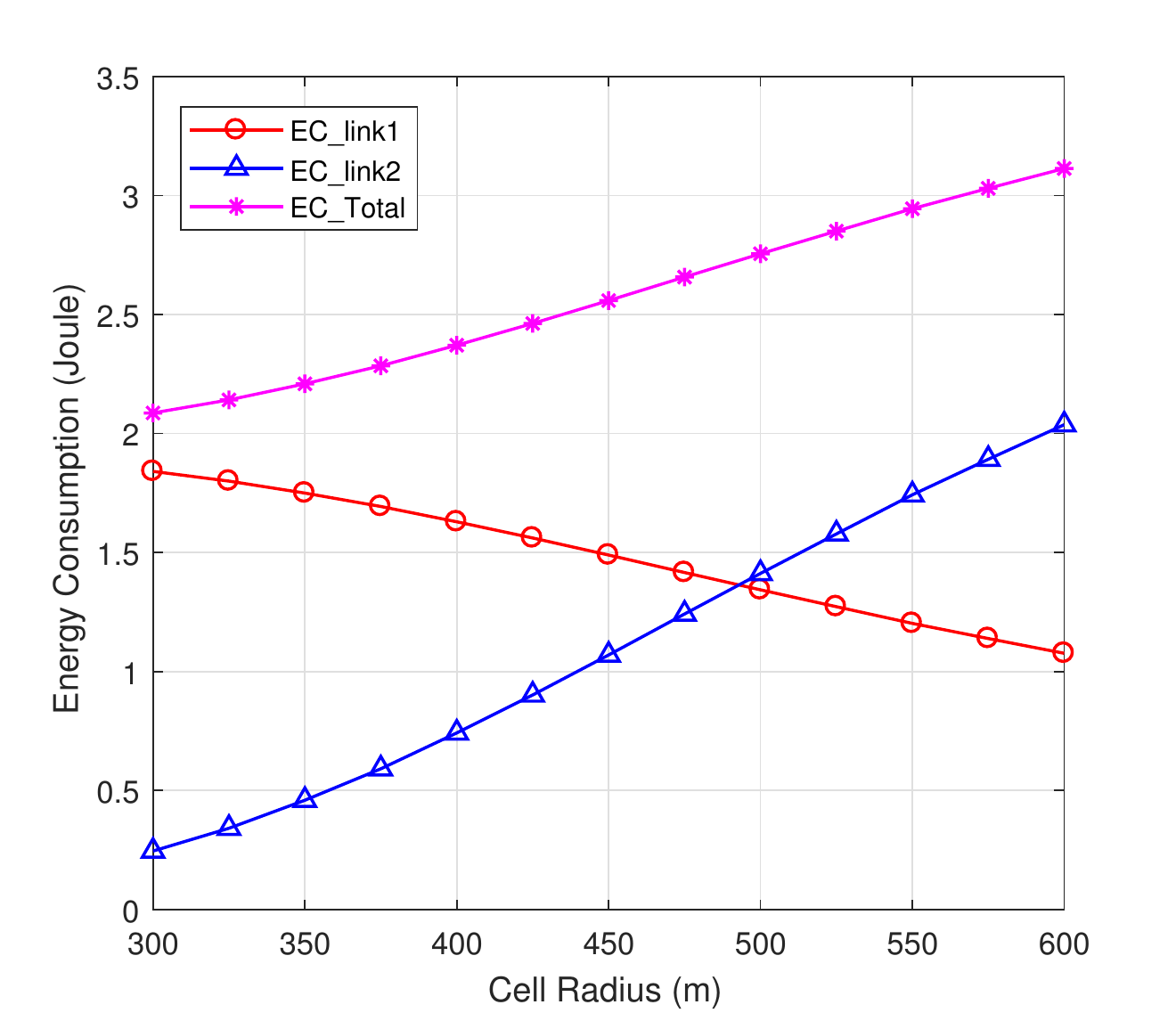}\\
	\caption{The energy consumption versus outage probability $\varepsilon_o$ with different estimated error variances.  $L=3.2\times10^7$ bits, $C=10^3\times[1,1]$, $\kappa=10^{-28}\times [0.8, 1.2]$, $f=10^9\times[0.8,1]$, $T_{\max}=0.01$ s, $P_{\max}=0.01$ W and $\varepsilon_o=0.1$.} \label{Fig9}
\end{minipage}
\begin{minipage}[t]{0.48\textwidth}
\centering
\graphicspath{{./figures/}}	\includegraphics[width=0.95\linewidth]{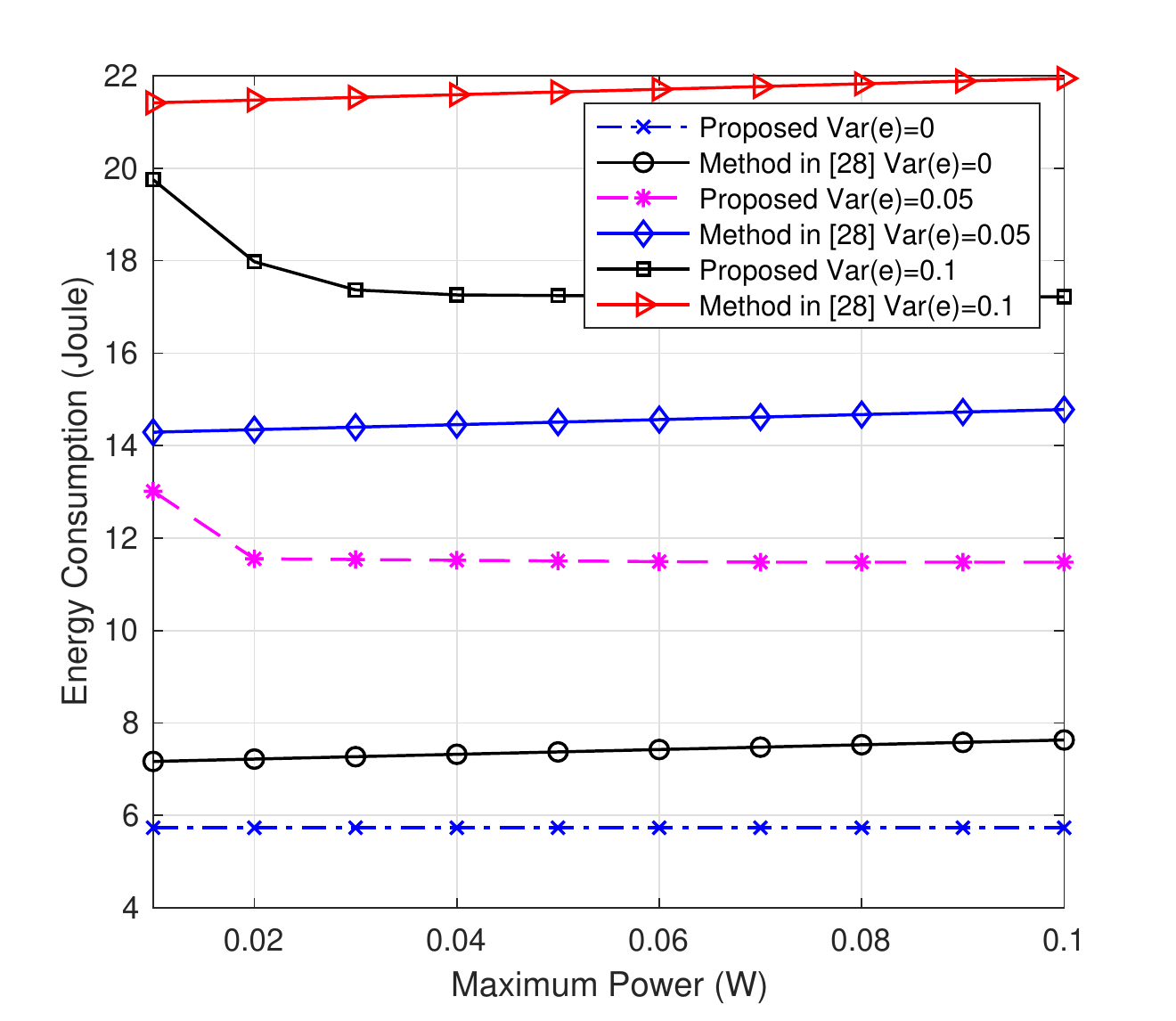}\\
	\caption{Energy consumption comparison with the method in \cite{TXTranTVT2018} (three users and three BSs). Radius: 500 m, $L=1\times10^8$ bits, $C=10^3\times[1,1,1]$, $\kappa=10^{-28}\times [0.8, 1.2,0.8]$, $f=10^9\times[0.8,1,0.8]$ and $T_{\max}=0.1 s$.} \label{Fig12}
\end{minipage}
\end{figure}
Fig. \ref{Fig9} depicts the energy consumption versus cell Radius. It can be observed that the energy consumption increases as the cell radius increases. This is because larger distances from the user to BSs will lead to higher energy consumption to meet QoS requirement due to poor channels. From this figure, it can also be seen that the energy consumption of Link 1 decreases while that of Link 2 increases.



Fig. \ref{Fig12} shows the performance comparison of the proposed scheme in this work with the power allocation scheme proposed in \cite{TXTranTVT2018}. From this figure, we can see that the proposed scheme in this work can provide lower energy consumption than the proposed power allocation scheme in \cite{TXTranTVT2018}. This is because that the proposed scheme in this paper is an optimal solution, while the bi-section search power allocation scheme in \cite{TXTranTVT2018} is a suboptimal solution by using an approximation to transfer the non-convex into a convex problem. It also shows that the proposed scheme in this work is more energy efficient than that in \cite{TXTranTVT2018} since the NOMA transmission is applied in this work.
\section{Conclusion}

In this paper, considering the imperfect CSI, we proposed an energy-efficient resource allocation scheme for a multi-user, multi-cell NOMA-MEC network. Specifically, based on the estimated channel model, we transformed the probabilistic problem into a nonprobabilistic one by incorporating the outage constraint into the objective function. Subsequently, we derived the closed-form expressions of the task and power allocation were derived to achieve the minimum energy consumption for the two-BS case. The analysis of the optimal solution provides the conditions and optimal solutions of two transmission models: 1. OMA offloading transmission to BSs; 2. Pure NOMA offloading transmission to both BSs. In addition, a low complexity user association algorithm was proposed for the multi-user and multi-BS case. Simulation results showed that the proposed solution can achieve better performance than the conventional OMA system, and the proposed user association algorithm can achieve lower complexity than the exhaustive search method. Since the closed-form solution can largely reduce the computation complexity, and the proposed scheme can be efficiently implemented in practice.

\appendix
\subsection{Proof of Proposition \ref{Out_RR}}\label{ProofProp1}
 	By using \eqref{C_m} and \eqref{equalOut}, we have 
\begin{equation}\label{CDF}
\begin{aligned}
&\Pr\left[B\log_{2}\left(1+\frac{G_{m}p_{m}}{G_{m}\sum\limits _{i=m+1}^{M}p_{i}+1}\right)<R_{m}|\hat{G}_{m}\right]\\
=&\Pr\left[G_{m}p_{m}<\gamma_{s}\left(G_{m}\sum\limits _{i=m+1}^{M}p_{i}+1\right)|\hat{G}_{m}\right]\\
=&\Pr\left[G_{m}<\frac{\gamma_{s}}{p_{m}-\gamma_{s}\sum\limits _{i=m+1}^{M}p_{i}}|\hat{G}_{m}\right]\\
=&F_{{G}_m}\left(\frac{\gamma_{s}}{p_{m}-\gamma_{s}\sum\limits _{i=m+1}^{M}p_{i}}|\hat{G}_{m}\right)
\end{aligned}
\end{equation}
where $F_{{G}_m}(x)$ denotes the cumulative distribution function (CDF) of $G_m$ and $\gamma_{s}=2^{\frac{R_{m}}{B}}-1$. Since $G_m$ is a non-central Chi-square distributed random variable with degrees of freedom 2. The non-centrality parameter is $\lambda_m^2=\frac{\hat{G}_m\sigma_{z}^2d_m^{\alpha}}{\sigma_{\epsilon}^2/2}$. Thus we have
\begin{equation}
F_{{G}_m}(x)
=1-Q_1\left(\lambda_m,\sqrt{\frac{\sigma_{z}^2d_m^{\alpha}x}{\sigma_{\epsilon}^2/2}}\right)=\varepsilon_o
\end{equation}
where $Q_1(a,b)=\exp\left(-\frac{a^2+b^2}{2}\right)\sum\limits_{k=0}^{\infty}(\frac{a}{b})^kI_k(ab)$ is the Marcum-Q-function $Q_1$ with modified Bessel function $I_k$ of order $k$ \cite{SMImperfect2015}. 
The inverse CDF of $G_m$, $F_{{G}_m}^{-1}(\varepsilon_o)$, can be used to obtain $\hat{R}_m$. The inverse CDF can be evaluated by using a lookup table \cite{NgTVT2010}. To obtain a closed-form result, we approximate non-central chi-square distribution, $\chi_{\hat{G}_{m}}^{2}\left(\lambda_{m}\right)$, by central chi-square distribution expression, $\chi_{\hat{G}_{m}}^{2}\left(0\right)$,  \cite{SMImperfect2015}:

\begin{equation}
\Pr \left[\chi_{\hat{G}_{m}}^{2}\left(\lambda_{m}\right)<x \right] \thickapprox \Pr \left[\chi_{\hat{G}_{m}}^{2}\left(0\right)<\frac{x}{1+\lambda_{m}^2/2} \right]
\end{equation}
This approximation can be proved accurate when the ratio between the non-central parameter $\lambda_m^2$ and degrees of freedom 2 is less than 0.2, which is $\lambda_m^2/2 \leq 0.2$. It has been shown in \cite{PJCIIEEETSP2011} that this approximation works well even the estimated error is large. Since central chi-square distribution with freedom 2 is exponential distribution, this approximation can transform the probabilistic constraint into a deterministic closed-form
\begin{equation}\label{exp}
\Pr \left[\chi_{\hat{G}_{m}}^{2}\left(0\right)<\frac{x}{1+\lambda_{m}^2/2} \right]=1-\exp\left({-\frac{x}{2\left(1+\lambda_{m}^2/2\right)}}\right).
\end{equation}
According to \eqref{equalOut} and \eqref{exp}, we have  
\begin{equation}
1-\exp\left({-\frac{\frac{\sigma_{z}^2d_m^{\alpha}\gamma_{s}}{p_{m}-\gamma_{s}\sum\limits _{i=m+1}^{M}p_{i}}}{2(1+\lambda_{m}^2/2)}}\right)=\varepsilon_o.
\end{equation}
Finally, we obtain the target data rate
\begin{equation}
\gamma_{s}=\frac{H_mp_m}{1+H_m\sum\limits _{i=m+1}^{M}p_{i}}
\end{equation}
where $H_m=-\ln(1-\varepsilon_o)2(1+\lambda_{m}^2/2)/(\sigma_{z}^2d_m^{\alpha})$.

\subsection{The Optimal Power Derivation}\label{PowerDerivation}
The energy consumption function of $p_1$ and $p_2$ can be written by
\begin{equation} \label{eq:EC}
\begin{aligned}
	E(p_1,p_2 )=&\frac{\beta_1 Lp_1}{(1-\varepsilon_o)B\log_2(1+\frac{H_1p_1}{H_1p_2+1})}\\&+\frac{(1-\beta_1) Lp_2}{(1-\varepsilon_o)B\log_2(1+H_2p_2)}\\&+\kappa_{1}\beta_1LC_{1}f_{1}^{2}+\kappa_{2}(1-\beta_1)LC_{2}f_{2}^{2}.
\end{aligned}
\end{equation}
It is difficult to find the optimal solution of problem $g(\beta_1)$ due to the nonconvexity of the objective function. However, the problem can be equivalently transformed into a bilevel programming problem, with the upper-level variable $p_2$, which is given by
\begin{subequations}\label{Prob:P2}
	\begin{align}
	\underset{0\leq p_{1}\leq P_{\max}}{\min}h\left(p_{1}\right)&\triangleq\underset{p_{2}}{\min}\ E(p_{1},p_{2})\\
	\text{s.t.} \quad &0\leq p_{2}\leq P_{\max}-p_{1}, \\
	&\frac{(1-\beta_1)L}{(1-\varepsilon_o)B\log_2(1+H_2p_2)}\leq T_{\max}, \label{T1} \\
	&\frac{\beta_1L}{(1-\varepsilon_o)B\log_2(1+\frac{H_1p_1}{H_1p_2+1})}\leq T_{\max} \label{T2}\\
	\nonumber	
	\end{align}
\end{subequations}
where $h(p_1)$ is the inner optimization problem with respect to $p_1$. Constraints \eqref{T1} and \eqref{T2} can be rewritten by 
\begin{equation}
	\frac{\left(2^{A(1-\beta_1)}-1\right)}{H_{2}}
	\leq p_2 \leq \frac{1}{H_1}\left(\frac{H_1p_1}{2^{A\beta_1}-1}-1\right). \\
\end{equation}
Take the partial derivative of \eqref{eq:EC} with respect to $p_2$, we have
\begin{equation}\label{eq:p_2partial}
\begin{aligned}
	\frac{\partial E }{\partial p_2}=&\frac{\beta_1 Lp_1}{(1-\varepsilon_o)B}\frac{\frac{H_1^2p_1}{\left(1+\frac{H_1p_1}{H_1p_2+1}\right)\left(H_1P_2+1\right)^2\ln(2)}}{\left(\log_2\left(1+\frac{H_1p_1}{H_1p_2+1}\right)\right)^2}
	\\&+\frac{\beta_1 L}{(1-\varepsilon_o)B}\frac{\frac{(1+H_2p_2)\ln(1+H_2p_2)-H_2p_2}{(1+H_2p_2)\ln(2)}}{\left(\log_2(1+H_2p_2)\right)^2}.
\end{aligned}
\end{equation}
It can be observed that the first term and the second term in \eqref{eq:p_2partial} are positive. By using the inequality $x\ln x\geq x-1,\  \forall x>0$, we have $(1+H_2p_2)\ln(1+H_2p_2)-H_2p_2\geq 0$. Therefore, the partial derivative of \eqref{eq:p_2partial}, $\frac{\partial E}{\partial p_2}>0$. Therefore, it can be concluded that the energy consumption is monotonically increasing with $p_2$, and the optimal $p_2$ can be obtained at its the minimum, which is $p_2^*=\frac{2^{A(1-\beta_1)}-1}{H_{2}}$. Then $h(p_1)=E(p_1,p_2^*)$. Since $p_2^*$ is not a function of $p_1$, we can treat $p_2^*$ as a constant. The outer optimization problem can be written by
\begin{subequations}\label{Prob:P1}
	\begin{align}
	\underset{0\leq p_{1}\leq P_{\max}}{\min} &h\left(p_{1}\right)\\
	\text{s.t.} \quad &0\leq p_{1}\leq P_{\max} \\
	&\frac{\beta_1L}{(1-\varepsilon_o)B\log_2(1+\frac{H_1p_1}{H_1p_2^*+1})}\leq T_{\max} \label{T22}\\
	&\frac{(1-\beta_1)L}{(1-\varepsilon_o)B\log_2(1+H_2p_2^*)}\\&=\frac{\beta_1L}{(1-\varepsilon_o)B\log_2(1+\frac{H_1p_1}{H_1p_2^*+1})}.\\
	\nonumber	
	\end{align}
\end{subequations}
Take the partial derivative of \eqref{eq:EC} with respect to $p_1$, we have
\begin{equation}\label{eq:p_1$ 2 $partial}
\begin{aligned}
\frac{\partial h }{\partial p_1}=\frac{\beta_1 L}{(1-\varepsilon_o)B} \frac{\frac{\left(1+\frac{H_1p_1}{H_1p_2^*+1}\right)\ln\left(1+\frac{H_1p_1}{H_1p_2^*+1}\right)-\frac{H_1p_1}{H_1p_2^*+1}}{\left(1+\frac{H_1p_1}{H_1p_2^*+1}\right)\ln(2)}}{\left(\log_2\left(1+\frac{H_1p_1}{H_1p_2^*+1}\right)\right)^2}.
\end{aligned}
\end{equation}	
By using the inequality $x\ln x\geq x-1,\  \forall x>0$, we have 
\begin{equation}
	\left(1+\frac{H_1p_1}{H_1p_2+1}\right)\ln\left(1+\frac{H_1p_1}{H_1p_2+1}\right)\geq\frac{H_1p_1}{H_1p_2+1}.
\end{equation} 
It can be verified that $\frac{\partial E }{\partial p_1}\geq 0$. Therefore, it can be concluded that the energy consumption is monotonically increasing with $p_1$. We have the optimal solution of $p_1^*$, which is
\begin{equation} \label{p11optimal}
p_1^*=\left(2^{A\beta_1}-1\right)\left(\frac{1}{H_2}\left(2^{A\left(1-\beta_1\right)}-1\right)+\frac{1}{H_1}\right)
\end{equation}
Above all, the energy consumption is monotonically increasing with $p_1$ and $p_2$. Then the minimum energy consumption can be achieved when transmit power equals the minimum required power by the constraints in \eqref{Prob:Bilevel}.


\subsection{Optimal Solution Derivation}\label{OptimalSolution}
According to Lemma \ref{feasible}, the feasible range requirement of problem \eqref{Prob:E_min} is $2^A\leq 1+H_2P_{\max}$.

\emph{Case 1}: $g(\beta_1)$ decreases within its feasible region with its upper bound 1, which requires $\hat{\beta}_{1,\max}\geq 1$ and $\hat{\beta}_1\geq 1$. Then we have the conditions \eqref{Case1Con} for $\beta_1^*=1$.

\emph{Case 2}: $g(\beta_1)$ decrease within its feasible region with its upper bound $\hat{\beta}_{1,\max}$, which requires $0<\hat{\beta}_{1,\max}<1$ and $\hat{\beta}_1>\hat{\beta}_{1,\max}$. Then we have conditions \eqref{Case2Con} for $\beta_1^*=\hat{\beta}_{1,\max}$.

\emph{Case 3}: $g(\beta_1)$ first decreases until $\hat{\beta}_1$ and then increases. Its feasible region is upper bounded by 1, which requires $\hat{\beta}_{1,\max}\geq 1$ and $0<\hat{\beta}_1<1$. Then we have conditions \eqref{Case3Con} for $\beta_1^*=\hat{\beta}_1$.

\emph{Case 4}: $g(\beta_1)$ first decreases until $\hat{\beta}_1$ and then increases. Its feasible region is upper bounded by $\hat{\beta}_{1,\max} $, which requires $\hat{\beta}_{1,\max}< 1$ and $0<\hat{\beta}_1<\hat{\beta}_{1,\max}$. Then we have conditions \eqref{Case4Con} for $\beta_1^*=\hat{\beta}_1$.

\emph{Case 5}: $g(\beta_1)$ increases within its feasible region with its lower bound 0, which requires $\hat{\beta}_{1}\geq 0$. Combined with the feasible condition \eqref{feasibleCon}, then we have the conditions \eqref{Case1Con} for $\beta_1^*=1$.

\bibliographystyle{IEEEtran}

\begin{thebibliography}{10}
\providecommand{\url}[1]{#1}
\csname url@samestyle\endcsname
\providecommand{\newblock}{\relax}
\providecommand{\bibinfo}[2]{#2}
\providecommand{\BIBentrySTDinterwordspacing}{\spaceskip=0pt\relax}
\providecommand{\BIBentryALTinterwordstretchfactor}{4}
\providecommand{\BIBentryALTinterwordspacing}{\spaceskip=\fontdimen2\font plus
\BIBentryALTinterwordstretchfactor\fontdimen3\font minus
  \fontdimen4\font\relax}
\providecommand{\BIBforeignlanguage}[2]{{%
\expandafter\ifx\csname l@#1\endcsname\relax
\typeout{** WARNING: IEEEtran.bst: No hyphenation pattern has been}%
\typeout{** loaded for the language `#1'. Using the pattern for}%
\typeout{** the default language instead.}%
\else
\language=\csname l@#1\endcsname
\fi
#2}}
\providecommand{\BIBdecl}{\relax}
\BIBdecl

\bibitem{ACTutMEC2017}
A.~C. {Baktir}, A.~{Ozgovde}, and C.~{Ersoy}, ``How can edge computing benefit
  from software-defined networking: {A} survey, use cases, and future
  directions,'' \emph{IEEE Commun. Surveys Tuts.}, vol.~19, no.~4, pp.
  2359--2391, 4th quarter 2017.

\bibitem{TTMEC2017}
T.~{Taleb}, K.~{Samdanis}, B.~{Mada}, H.~{Flinck}, S.~{Dutta}, and
  D.~{Sabella}, ``On multi-access edge computing: {A} survey of the emerging
  {5G} network edge cloud architecture and orchestration,'' \emph{IEEE Commun.
  Surveys Tuts.}, vol.~19, no.~3, pp. 1657--1681, 3rd quarter 2017.

\bibitem{YMaoMECSureveys2017}
Y.~Mao, C.~You, J.~Zhang, K.~Huang, and K.~B. Letaief, ``A survey on mobile
  edge computing: {T}he communication perspective,'' \emph{IEEE Commun. Surveys
  Tuts.}, vol.~19, no.~4, pp. 2322--2358, 4th quarter 2017.

\bibitem{PMMEC2017}
P.~{Mach} and Z.~{Becvar}, ``Mobile edge computing: {A} survey on architecture
  and computation offloading,'' \emph{IEEE Commun. Surveys Tuts.}, vol.~19,
  no.~3, pp. 1628--1656, 3rd quarter 2017.

\bibitem{ZDing2018TCOM}
Z.~{Ding}, P.~{Fan}, and H.~V. {Poor}, ``Impact of non-orthogonal multiple
  access on the offloading of mobile edge computing,'' \emph{IEEE Tran.
  Commun.}, vol.~67, no.~1, pp. 375--390, Jan. 2019.

\bibitem{ZNingNOMAMEC2019TVT}
Z.~{Ning}, X.~{Wang}, and J.~{Huang}, ``Mobile edge computing-enabled {5G}
  vehicular networks: {T}oward the integration of communication and
  computing,'' \emph{IEEE Veh. Technol. Mag.}, vol.~14, no.~1, pp. 54--61, Mar.
  2019.

\bibitem{MVaeziNOMA2019}
M.~Vaezi, G.~Amarasuriya, Y.~Liu, F.~Fang, and Z.~Ding, ``Interplay between
  {NOMA} and other emerging technologies: {A} survey,'' \emph{IEEE Trans. on
  Cogn. Commun. and Netw.}, vol.~5, no.~4, pp. 900--919, Dec. 2019.

\bibitem{SBMECMag2014}
S.~{Barbarossa}, S.~{Sardellitti}, and P.~{Di Lorenzo}, ``Communicating while
  computing: {D}istributed mobile cloud computing over {5G} heterogeneous
  networks,'' \emph{IEEE Signal Process. Mag.}, vol.~31, no.~6, pp. 45--55,
  Nov. 2014.

\bibitem{PMMECMag2017}
P.~{Mach} and Z.~{Becvar}, ``Mobile edge computing: {A} survey on architecture
  and computation offloading,'' \emph{IEEE Commun. Surveys Tuts.}, vol.~19,
  no.~3, pp. 1628--1656, 3rd quarter 2017.

\bibitem{JRMECTWC2018}
J.~{Ren}, G.~{Yu}, Y.~{Cai}, and Y.~{He}, ``Latency optimization for resource
  allocation in mobile-edge computation offloading,'' \emph{IEEE Trans.
  Wireless Commun.}, vol.~17, no.~8, pp. 5506--5519, Aug. 2018.

\bibitem{YWangMEC2016}
Y.~{Wang}, M.~{Sheng}, X.~{Wang}, L.~{Wang}, and J.~{Li}, ``Mobile-edge
  computing: Partial computation offloading using dynamic voltage scaling,''
  \emph{IEEE Trans. Commun.}, vol.~64, no.~10, pp. 4268--4282, Oct. 2016.

\bibitem{SAMECSurvey2014}
S.~Abolfazli, Z.~Sanaei, E.~Ahmed, A.~Gani, and R.~Buyya, ``Cloud-based
  augmentation for mobile devices: {M}otivation, taxonomies, and open
  challenges,'' \emph{IEEE Commun. Surveys Tuts.}, vol.~16, no.~1, pp.
  337--368, 1st quarter 2014.

\bibitem{WShiIEEEIoT2016}
W.~Shi, J.~Cao, Q.~Zhang, Y.~Li, and L.~Xu, ``Edge computing: {V}ision and
  challenges,'' \emph{IEEE Internet Things J.}, vol.~3, no.~5, pp. 637--646,
  Oct. 2016.

\bibitem{ZDing2018WCLDely}
Z.~Ding, D.~W.~K. Ng, R.~Schober, and H.~V. Poor, ``Delay minimization for
  {NOMA-MEC} offloading,'' \emph{IEEE Signal Process. Lett.}, vol.~25, no.~12,
  pp. 1875--1879, Dec. 2018.

\bibitem{YWuTVTMECNOMA2019}
Y.~{Wu}, K.~{Ni}, C.~{Zhang}, L.~P. {Qian}, and D.~H.~K. {Tsang},
  ``{NOMA}-assisted multi-access mobile edge computing: {A} joint optimization
  of computation offloading and time allocation,'' \emph{IEEE Trans. Veh.
  Technol.}, vol.~67, no.~12, pp. 12\,244--12\,258, Dec. 2018.

\bibitem{FangTWC2019}
F.~Fang, Y.~Xu, C.~S. Z.~Ding, M.~Peng, and G.~K. Karagiannidis, ``Optimal
  resource allocation for delay minimization in {NOMA-MEC} networks,''
  \emph{IEEE Trans. Commun.}, pp. 1--1, 2020 (Early Access).

\bibitem{LQ2019Iot2019}
L.~P. {Qian}, A.~{Feng}, Y.~{Huang}, Y.~{Wu}, B.~{Ji}, and Z.~{Shi}, ``Optimal
  sic ordering and computation resource allocation in {MEC-A}ware {NOMA NB-IoT}
  networks,'' \emph{IEEE Internet Things}, vol.~6, no.~2, pp. 2806--2816, Apr.
  2019.

\bibitem{ZDing2018TVT}
Z.~{Ding}, J.~{Xu}, O.~A. {Dobre}, and V.~{Poor}, ``Joint power and time
  allocation for {NOMA-MEC} offloading,'' \emph{IEEE Trans. Veh. Technol.},
  vol.~68, no.~6, pp. 6207--6211, June 2019.

\bibitem{FWang2018TCOM}
F.~{Wang}, J.~{Xu}, and Z.~{Ding}, ``Multi-antenna noma for computation
  offloading in multiuser mobile edge computing systems,'' \emph{IEEE Trans.
  Commun.}, vol.~67, no.~3, pp. 2450--2463, Mar. 2019.

\bibitem{AKiani2018JIOT}
A.~Kiani and N.~Ansari, ``Edge computing aware {NOMA} for 5{G} networks,''
  \emph{IEEE Internet Things}, vol.~5, no.~2, pp. 1299--1306, Apr. 2018.

\bibitem{YPanCL2018}
Y.~Pan, M.~Chen, Z.~Yang, N.~Huang, and M.~Shikh-Bahaei, ``Energy efficient
  {NOMA}-based mobile edge computing offloading,'' \emph{IEEE Commun. Lett.},
  vol.~23, no.~2, pp. 310--313, Feb. 2019.

\bibitem{SHanIoT2019}
S.~{Han}, X.~{Xu}, S.~{Fang}, Y.~{Sun}, Y.~{Cao}, X.~{Tao}, and P.~{Zhang},
  ``Energy efficient secure computation offloading in {NOMA}-based m{MTC}
  networks for {IoT},'' \emph{IEEE Internet Things}, pp. 1--1, 2019.

\bibitem{ZYEEGC2018}
Z.~{Yang}, J.~{Hou}, and M.~{Shikh-Bahaei}, ``Energy efficient resource
  allocation for mobile-edge computation networks with {NOMA},'' in \emph{IEEE
  Globecom Workshops (GC Wkshps)}, Dec. 2018, pp. 1--7.

\bibitem{ZSongCL2018}
Z.~{Song}, Y.~{Liu}, and X.~{Sun}, ``Joint radio and computational resource
  allocation for {NOMA}-based mobile edge computing in heterogeneous
  networks,'' \emph{IEEE Commun. Lett.}, vol.~22, no.~12, pp. 2559--2562, Dec.
  2018.

\bibitem{QGEEDL2018GC}
Q.~{Gu}, G.~{Wang}, J.~{Liu}, R.~{Fan}, D.~{Fan}, and Z.~{Zhong}, ``Optimal
  offloading with non-orthogonal multiple access in mobile edge computing,'' in
  \emph{2018 IEEE Global Communications Conference (GLOBECOM)}, Dec. 2018, pp.
  1--5.

\bibitem{FangGCWS2019}
F.~{Fang}, K.~{Wang}, and Z.~{Ding}, ``Optimal task assignment and power
  allocation for downlink {NOMA MEC} networks,'' in \emph{2019 IEEE Globecom
  Workshops (GC Wkshps)}, Dec. 2019 (Accepted).

\bibitem{XDiaoAccess2019D2D}
X.~{Diao}, J.~{Zheng}, Y.~{Wu}, and Y.~{Cai}, ``Joint computing resource,
  power, and channel allocations for {D2D}-assisted and {NOMA}-based mobile
  edge computing,'' \emph{IEEE Access}, vol.~7, pp. 9243--9257, 2019.

\bibitem{TXTranTVT2018}
T.~X. Tran and D.~Pompili, ``Joint task offloading and resource allocation for
  multi-server mobile-edge computing networks,'' \emph{IEEE Trans. Veh.
  Technol.}, vol.~68, no.~1, pp. 856--868, Jan. 2019.

\bibitem{NgTVT2010}
D.~W.~K. {Ng} and R.~{Schober}, ``Cross-layer scheduling for {OFDMA}
  amplify-and-forward relay networks,'' \emph{IEEE Trans. Veh. Technol.},
  vol.~59, no.~3, pp. 1443--1458, Mar. 2010.

\bibitem{CWangTVT2017}
C.~{Wang}, F.~R. {Yu}, C.~{Liang}, Q.~{Chen}, and L.~{Tang}, ``Joint
  computation offloading and interference management in wireless cellular
  networks with mobile edge computing,'' \emph{IEEE Trans. Veh. Technol.},
  vol.~66, no.~8, pp. 7432--7445, Aug. 2017.

\bibitem{YHeTVT2018}
Y.~{He}, N.~{Zhao}, and H.~{Yin}, ``Integrated networking, caching, and
  computing for connected vehicles: A deep reinforcement learning approach,''
  \emph{IEEE Trans. Veh. Technol.}, vol.~67, no.~1, pp. 44--55, Jan. 2018.

\bibitem{YMaoJSAC2018}
Y.~{Mao}, J.~{Zhang}, and K.~B. {Letaief}, ``Dynamic computation offloading for
  mobile-edge computing with energy harvesting devices,'' \emph{IEEE J. Sel.
  Areas Commun.}, vol.~34, no.~12, pp. 3590--3605, Dec. 2016.

\bibitem{SMImperfect2015}
S.~{Mallick}, R.~{Devarajan}, R.~A. {Loodaricheh}, and V.~K. {Bhargava},
  ``Robust resource optimization for cooperative cognitive radio networks with
  imperfect {CSI},'' \emph{IEEE Trans. Wireless Commun.}, vol.~14, no.~2, pp.
  907--920, Feb. 2015.

\bibitem{Approxi}
D.~R. Cox and N.~Reid, ``Approximation to non-central distributions,''
  \emph{Can. J. Statist.}, vol.~15, no.~2, pp. 105--114, June 1987.

\bibitem{YXuTSP2017}
Y.~{Xu}, C.~{Shen}, Z.~{Ding}, X.~{Sun}, S.~{Yan}, G.~{Zhu}, and Z.~{Zhong},
  ``Joint beamforming and power-splitting control in downlink cooperative
  {SWIPT NOMA} systems,'' \emph{IEEE Trans. Signal Process.}, vol.~65, no.~18,
  pp. 4874--4886, Sept. 2017.

\bibitem{3GPP}
``{Study on downlink multiuser supersition transmission ({MUST}) for {LTE}
  (Release 13)},'' 3GPP TR, Tech. Rep. 36.859, Tech. Rep., Mar. 2015.

\bibitem{FangJSAC17}
F.~Fang, H.~Zhang, J.~Cheng, S.~Roy, and V.~C.~M. Leung, ``Joint user
  scheduling and power allocation optimization for energy-efficient {NOMA}
  systems with imperfect {CSI},'' \emph{IEEE J. Sel. Areas Commun.}, vol.~35,
  no.~12, pp. 2874--2885, Dec. 2017.

\bibitem{PJCIIEEETSP2011}
P.~{Chung}, H.~{Du}, and J.~{Gondzio}, ``A probabilistic constraint approach
  for robust transmit beamforming with imperfect channel information,''
  \emph{IEEE Trans. Signal Process.}, vol.~59, no.~6, pp. 2773--2782, June
  2011.

\end{thebibliography}

\end{document}